\documentclass[fleqn,10pt]{wlscirep}
\usepackage[utf8]{inputenc}
\usepackage[T1]{fontenc}
\usepackage{amsmath}
\usepackage{epsfig}
\usepackage{amsfonts}
\usepackage{amssymb}
\usepackage{float}
\usepackage{amsthm}
\usepackage{verbatim}
\usepackage{geometry}
\usepackage[numbers]{natbib}
\usepackage{graphicx}
\usepackage{subfig}
\usepackage{algorithm}
\usepackage{algpseudocode}
\usepackage{verbatim}
\usepackage{geometry}
\usepackage{enumitem}
\numberwithin{equation}{section}
\usepackage[export]{adjustbox}
\graphicspath{ {./images/} }
\newtheorem{theorem}{Theorem}[section]

\title{Modeling Cost-Associated Cooperation: A Dilemma of Species Interaction Unveiling New Aspects of Fear Effect}
\author[1]{Suvranil Chowdhury}
\author[2]{Susmita Sarkar}
\author[1,*]{Joydev Chattopadhyay}
\affil[1]{Indian Statistical Institute, Agricultural and Ecological Research Unit, 203, B.T. Road, Kolkata, 700108, India}
\affil[2]{University of Calcutta, Department of Applied Mathematics, 92, Acharya Prafulla Chandra Road, Kolkata, 700009, India}
\affil[*]{joydev@isical.ac.in}
\begin{abstract}
With limited resources, competition is widespread, yet cooperation persists across taxa, from microorganisms to large mammals. Recent observations reveal that contingent factors often drive cooperative interactions, with the intensity heterogeneously distributed within species. While cooperation has beneficial outcomes, it may also incur significant costs, largely depending on species density. This creates a dilemma that is pivotal in shaping sustainable cooperation strategies. Understanding how cooperation intensity governs the cost-benefit balance, and whether an optimal strategy exists for species survival, is a fundamental question in ecological research, and the focus of this study. We develop a novel mathematical model within Lotka-Volterra framework to explore the dynamics of cost-associated partial cooperation, which remains relatively unexplored in ODE model-based studies. Our findings demonstrate that partial cooperation benefits ecosystems up to a certain intensity, beyond which costs become dominant, leading to system collapse via heteroclinic bifurcation. This outcome captures the precise cost-cooperation dilemma, providing insights for adopting sustainable strategies and resource management for species survival. We propose a novel mathematical approach to detect and track heteroclinic orbits in predator-prey systems. Moreover, we show that introducing fear of predation can protect the regime shift, even with a type-I functional response, challenging traditional ecological views. Although fear is known to resolve the "paradox of enrichment," our results suggest that certain levels of partial cooperation can reestablish this dynamic even at higher fear intensity. Finally, we validate the system's dynamical robustness across functional responses through structural sensitivity analysis.
\end{abstract}
\begin{document}
\flushbottom
\maketitle
%\linenumbers
\section{Introduction}\label{introduction}
In the realm of ecosystem studies, stability serves as a cornerstone for species survival, commanding global attention from scientists. However, recent research emphasizes the necessity of investigating the underlying factors influencing stability or instability within ecological systems, moving beyond solely identifying conditions for stability or other dynamical properties. This dual pursuit not only enhances our comprehension of ecosystems but also offers crucial insights for managing fluctuations or destabilization by choosing the necessary methodology to lower the possibility of species extinction. Effective management often requires considering the combined effects of multiple biological phenomena rather than isolated factors. Even to comprehensively understand the factors influencing simple prey-predator interactions, contemporary studies emphasize the importance of exploring psychological, behavioral, and cognitive characteristics alongside direct interactions. \citep{le2014short}. Such traits exert significant influence over predator and prey populations, often outweighing the importance of direct interactions alone \citep{wang2016studying, zanette2019ecology}. In recent studies, fear of predation has gained central focus because it significantly influences prey behavior and psychology, impacting reproduction, competition, foraging, cooperation, and ultimately prey mortality \citep{zanette2011perceived,zanette2019ecology}. Recent work addresses various aspects of predation fear and emphasizes its stabilizing effect \citep{panday2018stability, pal2015stability, hossain2020fear}. For further insights, refer to recent reviews on fear-related studies \citep{he2022stability, sha2019backward, panday2021dynamics}.\\
Like predation fear, cooperation serves as a fundamental and crucial mechanism for enhancing species' survival chances. Species may cooperate to reduce predation risk, strengthen anti-predator defenses, or improve foraging efficiency. However, cooperation is widespread and influences many survival strategies \citep{muller2005conflict, dugatkin1997cooperation, reluga2005simulated, courchamp2008Allee}. These strategies include optimal resource use, increased vigilance, resource sharing, and joint efforts in raising offspring. Extensive research underscores the pivotal role of prey cooperation in shaping predator-prey dynamics, impacting both prey and predator behavior \citep{boucher1982ecology, bronstein2015mutualism}. But, a long-debated question among ecologists is whether the intensity of cooperation is uniformly distributed among individuals and whether this cooperation always yields beneficial outcomes only. Experiments involving food exchange and other survival-related activities have shown that, except for a few primates, most species exhibit contingent cooperative behavior \citep{melis2006engineering, de2000payment, rutte2008influence}. For example, female baboons avoid females in conflicts with kin \citep{cheney1999recognition}, and female chickadees assess their mate's dominance status by observing singing contests with neighboring males \citep{mennill2002female}. Such selective cooperation can benefit related individuals but may reduce reproduction rates across the population and can indirectly affect the intensity of inter-specific competition, risking extinction, especially in smaller groups \citep{smith1964group,martinez2017sexual,morrow2004sexual}. This heterogeneity challenges the assumption of uniformly distributed cooperative interaction among individuals \citep{silk2007strategic, de2010prosocial, silk2009nepotistic}.
Herd behavior and signal calling, often seen as cooperative traits, improve foraging by facilitating information flow among individuals. However, large groups may aid predators in locating prey, while smaller groups are more vulnerable to predation \citep{rubenstein2010cooperation, sherman1977nepotism}. Again cooperative breeding raises questions about its evolutionary advantages and mechanisms, despite potential fitness costs and extinction risks (as seen in Suricata suricatta) \citep{trivers1971evolution, clutton1998costs}. These dynamics evidently reflect the density dependence of the costs associated with cooperation.\\ 
Rather than uniformly distributed cooperation among individuals within a species, Cheney et al. \citep{cheney2011extent} argue that cooperation is a psychological, emotional, and sometimes cognitive strategy decision made by considering potential cost-benefit ratios, leading to what is termed "Partial Cooperation" \citep{stark2010dilemmas}. While cost-benefit conflicts of cooperation are often addressed through evolutionary game-theoretic models, the exploration of these impacts within a simple predator-prey ODE system remains largely unexplored. However, potential cost-benefit trade-offs within communities are crucial for understanding the evolutionary significance of cooperative strategies, yet they remain largely untapped areas in ODE-based prey-predator models, which are typically based on the assumption of homogeneous cooperative intensity and only beneficial interactive outcomes. Therefore, by addressing this research gap, this study primarily focuses on constructing an ODE predator-prey model capable of capturing the collective effect of cost-associated cooperative phenomena over the dynamics of predator-prey systems. The model will also aim to reflect the significant influence of population density on the cost-benefit trade-off that we addressed earlier. Addressing the cost-benefit trade-off in cooperation raises key questions that we aim to explore: how does the cost-benefit ratio shift with varying cooperation intensity, and how does this impact classical prey-predator dynamics? Moreover, is there an optimal cooperation intensity that maximizes ecosystem benefits? \\
Mutually beneficial interactions, such as cooperation can introduce destabilizing pressures at low population levels, resulting in Allee effects \citep{lidicker2010Allee,asmussen1979density}, and can alter the carrying capacity of populations, affecting the maximum sustainable size \citep{zhang2003mutualism,hamilton2009population,storch2019carrying}. Although these phenomena are often studied separately, it is well recognized that demographic changes occur simultaneously during intraspecific or interspecific interactions in nature—an area that, to the best of our knowledge, remains largely unaddressed in ODE-based studies. Hence, our exploration will seek to bridge this gap.\\
Fear of predation, as previously discussed, significantly influences anti-predator movement and other cooperative strategies driven by contingent behaviors in prey species in the presence of predators. Consequently, we integrate both cost-associated partial cooperation and fear of predation to more accurately capture the dynamics of predator-prey interactions and revisit the findings of Wang et al. \citep{wang2016modelling} at varying intensities of cooperation, which suggested that the effectiveness of fear is sensitive to the type of functional response term. This raises a critical question: Are our dynamical results robust, or do they depend on the structure of the functional response? Inspired by the work of Morozov et al. \citep{flora2011structural, adamson2013can, adamson2014defining, adamson2020identifying}, we will conduct a structural sensitivity analysis to evaluate the robustness of our findings concerning changes in the functional response term in section\eqref{structures}.\\
This study presents some novel mathematical insights that enhance our understanding of the dynamical results that we discussed in section\eqref{axial_math}. The inclusion of fractional powers in prey density transforms the growth equation into a transcendental form, increasing system complexity. We propose a simple yet intriguing mathematical scheme that facilitates the exploration of system dynamics and efficiently tackles these challenges. Moreover, detecting heteroclinic bifurcation is challenging, as the tracking of the heteroclinic orbit remains volatile, and standard methods for identifying these separatrices through trapping regions are difficult to estimate. But, the existence of heteroclinic bifurcation is crucial and significant in population biology. However, here we present a refined approach that accurately tracks orbit behavior near joining points of the two different separatrices, offering a precise method for detecting the heteroclinic curve and distinguishing it from the homoclinic one in section \eqref{heteroclinic_curve_sum}.
\section{Model formulation and biological rationale}\label{model formulation}
Understanding the impact of such trade-associated heterogeneous phenomena in a system is most conveniently done by studying its aggregated effect at the community level. In this work, we use the "degree of partial cooperation" to quantify the intensity of the aggregated outcome from this type of heterogeneous or resultant cost-associated cooperation. The term "degree of cooperation" is used in the article \citep{cheney2011extent} to describe the strength or intensity of cooperation among individuals, suggesting that it should be fractional or partially homogeneous. Mainly inspired by this definition, the term "degree of partial cooperation" will be used to quantify the level or intensity at which organisms cooperate, cluster, or group to enhance their fitness. This acknowledges that its effectiveness may be somewhat diminished compared to the classical assumption of cooperation, either due to associated costs or cognitive or selective cooperation approaches. But, the challenge remains in how to model it mathematically. \\
The impact and influence of predation fear on the survival strategy adopted by prey species remains a central focus of our study. Therefore, during the model formulation process,
 we consider the model proposed in \citep{wang2016modelling} as the foundational structure for the model we are going to develop which is
 \begin{align}\label{wangmodel_general}
      \frac{du}{dt}&= (\frac{r_0u}{1+Kv})-au^2-du-\phi(u)v,\\
      \frac{dv}{dt}&=c\phi(u)v-mv,\nonumber
\end{align}
where $ \phi(u) $ is the general functional response term, $ u $ represents the prey density and $ v $ represents predator density and the descriptions of the rest of the variables are mentioned in the Table\eqref{tab:parameters}.
 \begin{table}[h!]
     \centering
     \begin{tabular}{cc}
        $ r_0 $ = & birth rate of the prey species\\
        $K $ = &  label of fear  in prey species\\
        $a $ = &  label of intraspecific competition\\
        $d $ = &  natural death rate of the prey\\
         $\phi(u) $ =  &  general functional response term for the predator\\
          $ c $ =  &  conversion rate for the predator \\
          $ m $ =  & natural death rate of the predator\\
     \end{tabular}
    \caption{Description of the variables used in this study.}
     \label{tab:parameters}
 \end{table}
 To incorporate the aggregative effect of the heterogeneous strategy, we adopt a mean-field model approach inspired by the work \citep{ajraldi2011modeling,venturino2013spatiotemporal}, where the general model formulation structure was:
\begin{align*}
x'(t) &= ax\big(1-\frac{x}{K}\big)F(f, y) - dx - bx^2 - mx^\beta y^\theta ,\\
y'(t) &= nmx^\beta y^\theta - ey.
\end{align*}
Here $\beta$ and $\theta$ represent the intensity of group defense and mutual interference respectively. This approach, first introduced by Braza et al. \citep{braza2012predator} and later explored by Chattopadhyay et al. \citep{chattopadhyay2008patchy} in toxin-producing systems, provides a foundation for modeling cooperative dynamics. For additional insights into this concept, see the work of Mengxin He et al. \citep{he2022stability}. In all of these studies, authors explored how group leaving and anti-predator movement or herd behavior can severely affect predator-prey dynamics by mainly adopting the density-dependent power law approach to redefine the functional response and established that this approach can readily capture the aggregative dynamics at the individual level. But based on the preceding discussion in section\eqref{introduction}, it is clear that when prey species engage in cooperative strategies—whether for defense, offspring care, foraging efficiency, resource optimization, or fitness improvement—it directly or indirectly influences their per-capita reproduction rate \citep{hassell1975density,venturino2013spatiotemporal,xiao2019stability}. Additionally, the review also suggests that the effectiveness of cooperation-related costs and benefits, along with selective behavior intensity, is highly dependent on species density. These insights guide our study to explore the impact of cost-associated cooperative strategies through density-dependent prey reproduction rates, an area that remains largely unexplored but essential to understand. So this time our objective is to capture the collective impact of this cost-cooperation dilemma by studying the effective changes in prey reproduction. Hatton et al. \citep{hatton2015predator} analyzed 2260 global communities and found that prey per capita production follows a power law \(P = rB^k\), with \(k\) varying across trophic structures, where "r" denotes the prey production coefficient, \(P\) represents the rate of prey production, and \(B\) denotes the prey biomass. Inspired by this result and the preceding reviews, we will incorporate a power-law approach to introduce cost-associated cooperation by including the term "$x^N$" in the prey reproduction term, reducing the model \eqref{wangmodel_general} to
 % TO use & in the equation for aligning the lines in the equation correctly. 
 \begin{align}\label{main_model_1}
 \begin{cases}
     & \frac{du}{dt}= (\frac{r_0u^N}{1+Kv})-au^2-du-\phi(u)v, \\ 
     & \frac{dv}{dt}=c\phi(u)v-mv,
     \end{cases}
 \end{align}
  where this $N$ is termed as the degree of partial cooperation. Incorporating a quadratic term of population density in the growth equation of a species acknowledges that as cooperation increases (i.e., as population density increases), the positive impact on growth or reproduction accelerates, often leading to mutual benefits. However, in this context, neither the associated costs nor the cognitive, strategy-based heterogeneous cooperation are assumed and cooperation is assumed to be homogeneously distributed among individuals. Nevertheless, it has been established that cooperation or group living entails both costs and potential benefits. Hence, individuals must assess the cost-benefit ratio before adopting group living or cooperative behavior \citep{alexander1974evolution,rubenstein2010cooperation}. This demonstrates that the degree of partial cooperation effectively captures the essence of density-dependent cooperation strategies. Hence in our study the range of $N$ is taken as $[1,2)$ \\
 Interestingly for $ N=1 $ our model will exactly take the form of the model \eqref{wangmodel_general}. Hamilton's Rule provides valuable insights into the evolution of cooperation in species. Evaluating the genetic relatedness among individuals aids in assessing the costs and benefits of cooperative actions \citep{rubenstein2010cooperation}. This model suggests that species compute the fitness advantages, whether direct or indirect, of engaging in cooperative or altruistic behaviors by comparing the discounted benefits to the associated costs. Crucially, species begin to cooperate even when aware of the costs, only when the benefit-cost ratio exceeds 1 which means stable cooperation may only emerge when the benefits outweigh the costs within a community \citep{rubenstein2010cooperation,smith1964group}, the exact biological phenomenon we aim to reflect through the term "degree of partial cooperation". Therefore, our assumption of selecting $N>1$ appears biologically relevant. Taking \(N < 2\) but \(N > 1\) in our study signifies that the degree of partial cooperation increases the growth of prey species, though at a rate somewhat diminished compared to the classical cooperative approach due to the associated cost. Moreover, setting \(N \geq 2\) in \eqref{main_model_1} results in the system losing its boundedness property, as proven later. These findings collectively provide a justifiable biological and mathematical rationale for our assumption.
 \begin{figure}[h!]
    \begin{center}
   \includegraphics[height=3.0in,width=5.5in]{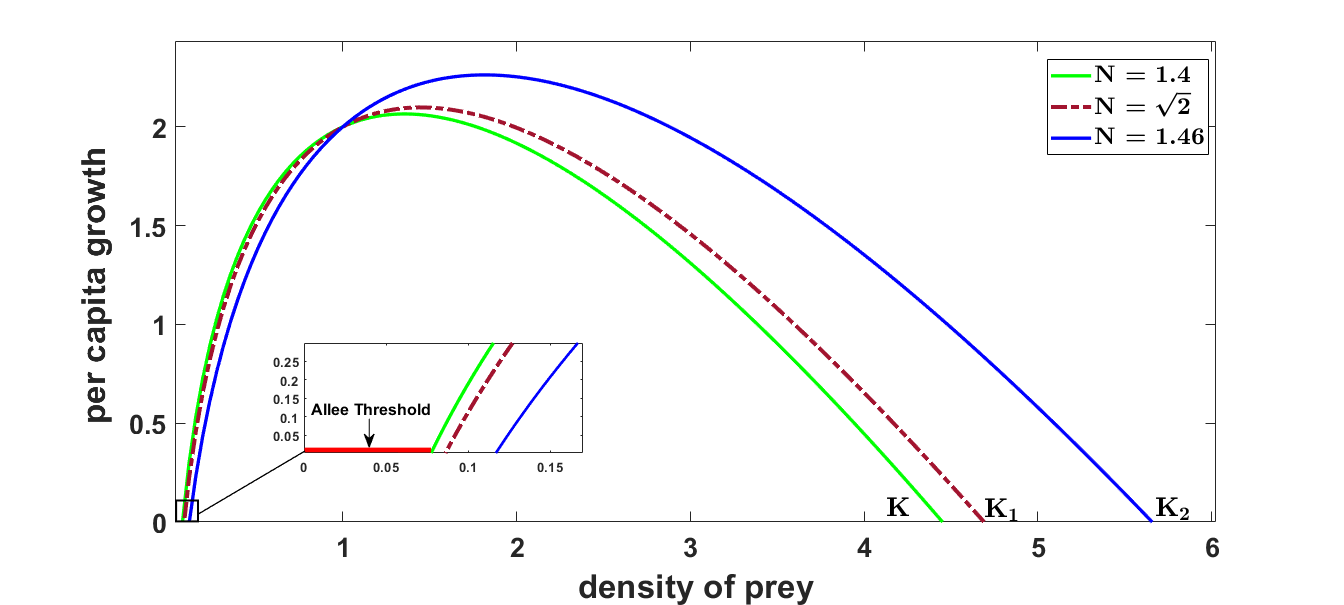}
    \end{center}
    \caption {This figure illustrates how the per capita growth rate changes with prey density in system \eqref{main_model_1} for various values of \(N\). As \(N\) increases from 1.4 (solid green curve) to \(\sqrt{2}\) (brown dashed curve) to 1.46 (solid blue curve), the system experiences a pronounced demographic Allee effect, with the threshold increasing concurrently (indicated by the bold red line). Additionally, as \(N\) rises, the system's carrying capacity (the right intersection of the per capita growth curve with the prey density axis) increases from \(K\) to \(K_1\) and then to \(K_2\). This suggests that prey can persist at higher densities with a higher degree of partial cooperation. The remaining parameter values are \(r_0 = 0.6\), \(a = 0.3\), and \(d = 0.1\).}
    \label{label fig 1}
    \end{figure}
 Throughout the paper, as demonstrated earlier, we are going to consider different types of functional response terms but let us first start with the linear functional response $ \phi(u)=pu $ where $ p $ represents the predation rate of the predator. So the main model turns into 
  \begin{align}\label{main_model_general}
 \begin{cases}
     & \frac{du}{dt}= \left(\frac{r_0u^\theta}{1+Kv}\right)u- au^2 - du - puv, \\ 
     & \frac{dv}{dt}=quv-mv,
     \\
     & \text{ where } q=cp \text{ and } 0\leq\theta<1.
     \end{cases}
 \end{align}
 Recent studies offer two definitions of carrying capacity: the classical definition and another that considers it as a stable equilibrium of diversity-dependent dynamics of species richness\citep{storch2019carrying}. In this work, we adopt the latter definition, defining the abscissa of the point $K$ as the carrying capacity for our system which we will discuss later.  
 \section{Possible Equilibria and The Criterion for Existence}
 In this section, we will discuss the condition of the existence of three types of equilibrium points. Which are,\\
 \begin{enumerate}
     \item Trivial Equilibrium  point =$(u_0,v_0)$ = $(0,0)$

     \item Semitrivial or axial Equilibrium points = $(\alpha ,0)$ 
           and $(\beta,0)$ 

     \item Nontrivial Equilibrium  = $(\Bar{u},\Bar{v})$, where ($\Bar{u},\Bar{v})$ satisfies the condition 
     \begin{equation}\label{nontrivial_fixed_condition}
    \begin{split}
           \Bar{u}=\frac{m}{q} 
 \hspace{0.4cm}\&  \hspace{0.4cm} 
        \left(\frac{r_0\Bar{u}^N}{1+K\Bar{v}}\right)-a\Bar{u}^2-d\Bar{u}-p\Bar{u}\Bar{v}=0. 
    \end{split}
    \end{equation}
 \end{enumerate}
While many studies opt for simplification by using reduced or linearized systems to bypass mathematical complexity, our approach is distinct. Instead of choosing such simplifications, we tackle the exact transcendental equations head-on, making the exact solution (axial equilibriums) unattainable. However, we develop a novel method to identify the conditions for root existence and leverage them to explore stability. This alternative approach proves both simple and ingenious.   
 \subsection*{Condition For the Existence of the Trivial Eqilibrium}
It is very easy to conclude that $(u_0,v_0) = (0,0)$ is the trivial equilibrium for the system and it will always exist in the system.
\subsection*{Condition For the Existence of the Axial Equilibria}\label{axial_math}
 However, due to the complexity of the model, it is challenging to find the exact form of the axial equilibria and the possible number of the equilibria. In fact, it doesn't seem easy to find their exact location. Hence, here we will use an alternative yet interesting way to overcome this problem. Here we excluded the case $N=1$ as in the case of $N=1$ the model \eqref{main_model_1} become identical with the work of \citep{wang2016modelling}
 \begin{theorem}
     For any real value of $N\in(1,2)$ the system will have exactly two axial equilibrium points one is in the left-hand part of the point $ (u^1,0) $ and the other one will be in the right-hand part.
 \end{theorem}
 Firstly, we are going to find the number of axial equilibria and then we will discuss their possible positions. Let us consider $v=0$ and from the first equation of the system \eqref{main_model_1} we have
\begin{align}\label{x_nullcline}
    f(u)= r_0u^{N-1}-au-d=0. 
\end{align}
Now the number of equilibria of the system can be quantified as the possible no of point of intersection of the prey nontrivial nullcline with the line $v=0$.\\
Here, $N \in \mathbb{R}$ and $0\leq N-1<1$. Now, we are going to discuss our study in two cases.
\subsubsection*{Case 1}
When $ N \in \mathbb{Q}$, in this case we can assume,
\begin{align*}
    N-1=\frac{p'}{q'},\textit{ where } p',q'\in\mathbb{Z}, q'\neq 0 \textit{ and } p'<q'.
\end{align*}
 After simplification and rearrangement of the terms in the standard form of a polynomial, then \eqref{x_nullcline} becomes,
\begin{align} \label{rationals_pow_poly}
   f(w) = aw^{q'}-r_0w^{p'}+d=0,\textit{ where, }
   w = u^\frac{1}{q'}
\end{align}
As there are two changes in the sign of the coefficients of the polynomial $f(w)$ then by Descartes's rule of sign, we can conclude that either the polynomial has two real positive roots or no real positive root. Again the curve \eqref{x_nullcline} is continuously differentiable, for $u>0$ and also we have,
\begin{align*}
   f'(u)
   % =& r_0 N u^{N -1}-a\\
   %=& r_0 (N-1) u^{N -2}-a
    =&0 ,\textit{ when }u=\Bigg(\frac{a}{r_0(N-1)}\Bigg)^\frac{1}{N-2} = u^{1}(\textit{say})\\
   \textit{Hence,} f''(u)=&r_0 (N-1)(N-2) u^{N -3}<0 , u\in (0,\infty)%[\textit{ as } N<2]
\textit{ which means }  f''(u^{1}) <0
\end{align*}
So, we can now conclude that the function will attain its maxima at $u=u^1$ and also the curve, represented by the equation \eqref{x_nullcline} is also concave downwards. Hence $ f(u^1)>0 $ also with the fact that $f(0)<0$ will ensure that the prey nontrivial nullcline can intersect the line $v=0$ exactly two distinct points $(\alpha,0)$ and $(\beta,0)$ respectively, laying two opposite sides of $(u^1,0)$ on the line $v=0$ with $\alpha<u^1<\beta$. Figure(\eqref{label fig 1}) clearly demonstrates this phenomenon. That means
\begin{align}
    %&\alpha<\Bigg(\frac{a}{r_0(N-1)}\Bigg)^\frac{1}{N-2}<\beta  \nonumber \\ 
    % \implies &\alpha^{N-2}>\Bigg(\frac{a}{r_0(N-1)}\Bigg)>\beta^{N-2}, \textit{ as, }1\leq N<2  \nonumber \\
    \alpha^{N-2}>\Bigg(\frac{a}{r_0(N-1)}&\Bigg)   \label{alpha_position}\\
    \textit{and, }&\Bigg(\frac{a}{r_0(N-1)}\Bigg)>\beta^{N-2}     \label{beta_position} %\textit{ as, }1\leq N<2 .
\end{align}
holds together. Now when $f(u^1)=0$ then the line $v=0$ will be a tangent to the curve \eqref{rationals_pow_poly}. Hence, there will be exactly one axial equilibrium. We will discuss this case later in the section on bifurcation analysis.
\subsubsection*{Case 2}
When $N \in \mathbb{R} - \mathbb{Q} $, then there always exists two rational numbers $ m,n $ (say),where $m<N<n$. Now if $ u>1,$
\begin{align*}
    r_0u^m-au-d 
    <r_0u^N-au-d  
    <r_0u^n-au-d 
\end{align*}
 and when $u\leq 1$ then,
\begin{align*}
     r_0u^m-au-d  
    \geq r_0u^N-au-d 
    \geq r_0u^n-au-d .
\end{align*}
Based on the previous findings, we can deduce that a curve with an irrational power will reside between two curves with rational powers. Figure(\eqref{label fig 1}) provides a clear graphical reference to illustrate this condition. As we have already proved in the case of rational values of $N$, we can establish that the curve with an irrational power will also intersect the same line at precisely two points. Hence we can now conclude that any real value of N with $1<N<2$ system will have two axial equilibria. Also, it is to be noted that the condition of maximality, in this case, will remain exactly the same as in the previous case, and hence the results \eqref{alpha_position} and \eqref{beta_position} will also be held for this case. These two inequalities are going to be pivotal for determining the stability of the axial equilibrium points.
\subsection*{Condition for the existence of the Interior Equilibrium Point}  
The nontrivial equilibrium point $(\Bar{u},\Bar{v})$ will satisfy the following conditions,
 \begin{align}\label{condition of non trivial equilibrium}
        \begin{cases}
           \Bar{u}=\frac{m}{q}\\
           \bigg(\frac{r_0\Bar{u}^N}{1+K\Bar{v}}\bigg)-a\Bar{u}^2-d\Bar{u}-p\Bar{u}\Bar{v}=0
         \end{cases}
 \end{align}
 that gives us,
 \begin{align}
 &(\frac{r_0\Bar{u}^N}{1+K\Bar{v}})-a\Bar{u}^2-d\Bar{u}- p\Bar{u}\Bar{v}=0\nonumber\\
 % \implies&
 % r_0\Bar{u}^N-a\Bar{u}^2-Ka\Bar{v}\Bar{u}^2-p\Bar{u}\Bar{v}-Kp\Bar{u}\Bar{v}^2-d\Bar{u}-Kd\Bar{u}\Bar{v}=0\nonumber\\
  % \implies&
 % \Bar{v} = \frac{-(aK\Bar{u}^2+p\Bar{u}+Kd\Bar{u})+\sqrt{(aK\Bar{u}^2+p\Bar{u}+Kd\Bar{u})^2+4pK\Bar{u}^2(r_0\Bar{u}^{N-1}-a\Bar{u}-d)}}{2pK\Bar{u}}\nonumber\\
 \textit{which means, }&
 \Bar{v} = \frac{-(aK\Bar{u}+p+Kd)+\sqrt{(aK\Bar{u}+p+Kd)^2+4pK(r_0\Bar{u}^{N-1}-a\Bar{u}-d)}}{2pK}\nonumber
  \end{align}
  %\subsection*{Theorem 4(Condition of Existence )}\label{condition of exsistence of interior fixed point}  
 So, the nontrivial equilibrium will exist if $\Bar{v}>0$, that means when,
  \begin{align*}%-(&aK\Bar{u}+p+Kd)+\sqrt{(aK\Bar{u}+p+Kd)^2+4pK(r_0\Bar{u}^{N-1}-a\Bar{u}-d)}> 0\\
  % \implies& (aK\Bar{u}+p+Kd)^2 <(aK\Bar{u}-p+Kd)^2+4pK(r_0\Bar{u}^{N-1})\\
 r_0\Bar{u}^{N-1}-a\Bar{u}-d>0.
  \end{align*}
  Based on the aforementioned condition, we can conclude that if $f(\Bar{u})>0$, then the nontrivial equilibrium will exist and as the nontrivial predator nullcline is a vertical line then it can intersect the nontrivial prey nullcline at exactly one point when $\alpha \leq \frac{m}{q} \leq \beta$, which means the system has exactly one nontrivial equilibrium point.\\
 %Interestingly this condition is just the opposite case of \eqref{stab_of_axial_N<2}. So the condition says that {the nontrivial equilibrium will be born when the axial equilibrium loses its stability}.\\
If $f(u_1)>0$ and $\alpha \leq \frac{m}{q} \leq \beta$ holds together then that will imply $f(\bar{u})>0$.\\
So, from the previous discussion, we may now conclude that if $f(u_1)>0$ and $\alpha \leq \frac{m}{q} \leq \beta$ then both the axial equilibria and the nontrivial equilibrium will exist.\\ 
Now as a conclusion of the above discussion, we present a theorem below.
\begin{theorem}
    Both the axial equilibria and the nontrivial equilibrium will co-exist in the system if $f(u_1)>0$ and $\alpha \leq \frac{m}{q} \leq \beta$.But if $ \frac{m}{q} \notin [\alpha,\beta]$ but $f(u_1)>0$ then only the axial equilibria will exist.
\end{theorem}        
\section{Stability Analysis of equilibria}
 In this section, we will examine the stability of equilibrium points by analyzing the sign of the eigenvalues of the Jacobian matrix. To begin, let us find the Jacobian matrix associated with equation \eqref{main_model_1}.
 \begin{align}\label{general_jacobian_1}
     J = \begin{bmatrix}
        J_{11}  &  J_{12}\\
        J_{21}  &  J_{22}
      \end{bmatrix}
 \end{align}
 where $J_{11}=(\frac{Nr_0u^{N-1}}{1+Kv})-2au-pv-d,J_{12}=(\frac{-Kr_0u^N}{(1+Kv)^2})-pu,J_{21}=qv  $ and $  J_{22}=qu-m $\\
 where it is assumed that $1<N<2$,
 %\subsection*{Existence of nontrivial equilibrium }\label{exist_non_triv}
 %\paragraph(HAVE TO CHECK)
 \label{partial cooperation}
 \subsection*{For trivial equilibrium point}
     \begin{theorem}
         The trivial equilibrium point is always stable.
     \end{theorem}
 \begin{proof}
      At the trivial equilibrium point $E_0 = (u_0,v_0)$, the jacobian \eqref{general_jacobian_1} will be,
  \begin{align}\label{trivial_jacobian_1.1}
     J_{(u_0,v_0)}=J_0 =
     \begin{bmatrix}
        -d  &  0\\
        0  &  -m
      \end{bmatrix}
 \end{align}\\
 Hence trace $ J_{0}=-d-m<0 $ and det $J_{0}=dm$ , that means one of the eigen values are negative .So we can say that the fixed point is stable asymptotically if m and d remain nonzero.  
 \end{proof}   
    \subsection*{For axial equilibrium points }%($u^{*},0$)}} 
    \begin{theorem}
       The axial equilibrium $(\alpha,0)$ will be, 
     \begin{enumerate}[label=(\alph*)]
         \item  a  saddle point if only if $\alpha<\frac{m}{q}$.
         \item an unstable node if and only if $\alpha>\frac{m}{q}$ 
         %\item $ u'=\frac{m}{q}$ then the fixed point will be a hyprbolic repeller.
     \end{enumerate}
     %a) a  saddle point if only if $u^*<\frac{m}{q}$.
    % b) an unstable node if and only if $u^*>\frac{m}{q}$ .
    % c) $ u^*=\frac{m}{q}$ then the fixed point will be a hyprbolic repeller.  
    \end{theorem}     
     \begin{proof}
         At the equilibrium point is $E_1=(\alpha,0)=(u^{*},0)$ (say), the Jacobean  matrix \eqref{general_jacobian_1} will take the form 
 \begin{align}\label{boundary_jacobian_2.1}
   J_{(\alpha,0)}=J_1=  \begin{bmatrix}
     Nr_0u^{*N-1}-2au^*-d  &  -Kr_0u^{*N}-pu^*\\
     0  &   qu^*-m
     \end{bmatrix}
     \end{align}   
 \text{and also we have, }     
 \begin{align}     
       r_0u^{*N-1}-au^*-d=0 \label{condition_axial1}
 \end{align}
 at the equilibrium point $E_1$.\\
 Now the eigenvalues of the jacobian are $\lambda_1=Nr_0u^{*N-1}-2au^*-d$ and $\lambda_2=qu^*-m$ respectively and also by using the results \eqref{alpha_position} and \eqref{condition_axial1}
 \begin{align}
 \lambda_1=& Nr_0u^{*N-1}-2au^*-d
          =Nr_0u^{*N-1}-au^*-r_0u^{*N-1}\nonumber \\
         % =&(N-1)ru^*u^{*(N-2)}-au^* \\
          >&(N-1)ru^*\frac{a}{r(N-1)}-au^* >0\nonumber
 \end{align}
 and the other eigen value say $\lambda_2<0 $ if and only if  $u^*<\frac{m}{q}$.\\
 Hence if  $u^*<\frac{m}{q}$ the axial fixed point will be a saddle point.\\
Similarly,$\lambda_2>0 $ if and only if  $u^*>\frac{m}{q}$, and at this time the fixed point will be an unstable node.\\
Now, if $u^*=\frac{m}{q}$ then $\lambda_2=0$ and $\lambda_1>0$, implies that the fixed point will be a hyperbolic repeller.\\
     \end{proof}     
 \begin{theorem}
     The axial equilibrium $(\beta,0)$ will be 
     \begin{enumerate}[label=(\alph*)]
         \item a stable node if only if $\beta<\frac{m}{q}$.
         \item an saddle node if  $\beta>\frac{m}{q}$.
         %\item if  $ u'=\frac{m}{q}$ then the fixed point will be a hyprbolic attractor.
     \end{enumerate}
 \end{theorem}
     \begin{proof}
   Here the equilibrium point is $E_2=(\beta,0)=(u',0)$ (say) and the Jacobean  matrix \eqref{general_jacobian_1} will take the form 
    \begin{align}\label{boundary2_jacobian_2.1}
   J_1=  \begin{bmatrix}
     Nr_0u'^{N-1}-2au'-d  &  -Kr_0u^{'N}-pu'\\
     0  &   qu'-m
     \end{bmatrix}
     \end{align}
     \text{and also we have, }     
 \begin{align}      
       r_0u'^{N-1}-au'-d=0 \label{condition_axial2} \textit{ and } \Bigg(\frac{a}{r_0(N-1)}\Bigg)>\beta^{N-2} 
 \end{align}
 at the equilibrium point $E_2$.\\
 Now the eigenvalues of the jacobian are $\lambda'_1=Nr_0u'^{N-1}-2au'-d$ and $\lambda'_2=qu'-m$ respectively and also using the results \eqref{beta_position} and \eqref{condition_axial2} we can have,
\begin{align}
 \lambda'_1=& Nr_0u'^{N-1}-2au'-d
          = Nr_0u'^{(N-1)}-au'-r_0u'^{(N-1)} \nonumber \\
          %=&(N-1)ru'u'^{(N-2)}-au' \nonumber\\
          <&(N-1)ru'\frac{a}{r(N-1)}-au' <0\nonumber
 \end{align}
 and the other eigen value is $\lambda'_2<0 $ if and only if  $u'<\frac{m}{q}$. Hence when the strict inequality holds, the predator-free equilibrium will be a stable node, and at the time of equality, the equilibrium will be a hyperbolic attractor. Similarly,$\lambda'_2>0 $ if and only if  $u'>\frac{m}{q}$, and at this time the fixed point will be an unstable fixed point it will be a saddle-node.
 \end{proof}
\subsection*{For the nontrivial equilibrium}
In the prior section, the nontrivial equilibrium point $(\bar{u},\bar{v})$ adheres to the condition \eqref{nontrivial_fixed_condition}. Additionally, if $f(u_1) > 0$ and the inequalities $\alpha<\frac{m}{q}<\beta$ are both satisfied, the system will exhibit a single nontrivial equilibrium. We will now delve into the stability analysis of this fixed point using the subsequent theorems.
  \begin{theorem}\label{condition of existence of interior fixed point} 
  The non-trivial equilibrium will be 
  \begin{enumerate}[label=(\alph*)]
    \item a stable focus if $N<1+\frac{a\bar{u}}{a\bar{u}+p\bar{v}+d}$,
    \item will be an unstable fixed point if $N>1+\frac{a\bar{u}}{a\bar{u}+p\bar{v}+d}$
  \end{enumerate}
\end{theorem}
 \begin{proof}
  At the point $\Bar{E} =(\Bar{u},\Bar{v})$ the jacobian \eqref{general_jacobian_1} will take the form, 
\begin{align}\label{nontrivial_jacobian}
J_{(\bar{u},\bar{v})}=J_{2}=\begin{bmatrix}
  (\frac{nr_0\Bar{u}^{N-1}}{1+K\Bar{v}})-2a\Bar{u}-p\Bar{v}-d & -\frac{Kr_0\Bar{u}^N}{(1+K\Bar{v})^2}-p\Bar{u}\\
  q\Bar{v}  &  0 
  \end{bmatrix}
\end{align}
  Now $detJ_{2}=(\frac{Kr_0\Bar{u}^N}{(1+K\Bar{v})^2}+p\Bar{u})(q\Bar{v})>0$ 
  and $Trace J_{2}=(\frac{Nr_0\Bar{u}^{N-1}}{1+K\Bar{v}})-2a\Bar{u}-p\Bar{v}-d$.
  After some simplification we have,
  \begin{align*}
  Trace J_{2}
    =(\frac{Nr_0\Bar{u}^{N-1}}{1+K\Bar{v}})-2a\Bar{u}-p\Bar{v}-d\\
%   =\frac{N}{\Bar{u}}(\frac{r_0 \Bar{u}^N}{1+K\Bar{v}})-2a\Bar{u}-p\Bar{v}-d\\
%   =\frac{N}{\Bar{u}}(a\Bar{u}^2+p\Bar{u}\Bar{v}+d\Bar{u})-2a\Bar{u}-p\Bar{v}-d
% \\
=(N-2)a\Bar{u}+(N-1)(p\Bar{v}+d) 
  \end{align*}
  So, $Trace J_{2}<0$ only when
  \begin{align*}
      &(2-N)a\Bar{u}>(N-1)(p\Bar{v}+d)
   \end{align*}
     which implies that,
     \begin{align}\label{condition_nontriv_stable_1} 
     N<1+\frac{a\bar{u}}{a\bar{u}+p\bar{v}+d}= 1+\frac{1}{1+\bigg(\frac{p\bar{v}}{a\bar{u}}\bigg)+\frac{d}{a\bar{u}}}
      =N'_0(\textit{say})   
  \end{align}
  Now to find the nature of the stability at $N=N'_0$ first, we take,
  \begin{align*}
       (Trace(J_{2}))^2-4det(J_{2})=0 \textit{ when, }
     (&(N-2)a\Bar{u}+(N-1)(p\Bar{v}+d) )^2=4det(J_2)\\
     % \implies (N-2)a\Bar{u}+(N-1)(p\Bar{v}+d)=&\sqrt{4det(J_2)}\\
     % \implies N=&\frac{\sqrt{4det(J_2)}+2a\bar{u}+p\bar{v}+d}{a\bar{u}+p\bar{v}+d}\\
    \textit{which implies, } N = &1+\frac{a\bar{u}+\sqrt{4det(J_2)}}{a\bar{u}+p\bar{v}+d}.
  \end{align*}
  Hence,$(Trace(J_{2}))^2-4det(J_{2})< 0$ if 
  \begin{align}\label{condition of focus}
  N<1+\frac{a\bar{u}+\sqrt{4det(J_2)}}{a\bar{u}+p\bar{v}+d}
  \end{align} 
  So, by comparing the conditions \ref{condition_nontriv_stable_1} and \ref{condition of focus} we can now conclude that when the interior equilibrium will be stable, that means when, 
  \begin{align*}
     N<N'_0=1+\frac{a\bar{u}}{a\bar{u}+p\bar{v}+d}<1+\frac{a\bar{u}+\sqrt{4det(J_2)}}{a\bar{u}+p\bar{v}+d} \label{condition_nontriv_stable_2}
  \end{align*}
 When stability conditions are met, the nature of the coexistence equilibrium is of focus type.
\end{proof} 
 \section{Bifurcation Analysis} 
  \subsection{Transcritical Bifurcation}
 We've established mathematically that the axial equilibrium is stable when \( m > q\beta \). However, at \( m = q\beta \), the predator nullcline intersects with the prey nullcline at the axial equilibrium \((\beta,0)\). As \( m \) decreases further, the axial equilibrium becomes unstable, and a new stable equilibrium emerges (the interior equilibrium). This phenomenon is termed the transcritical bifurcation in bifurcation theory. Here, we will provide a mathematical proof of this bifurcation using Sotomayor's Theorem \citep{perko2013differential} at the point \( P_1 = (\beta, 0) \).
 \begin{theorem}
 The axial equilibrium $(\beta,0)$ will go through a transcritical bifurcation at $m=q\beta$.
 \end{theorem}
\begin{proof}
For detailed proof see Appendix\eqref{APP_trans} 
\end{proof}
  \subsection{ Hopf bifurcation }\label{APPENDIX-E}  
The Hopf bifurcation is a two-dimensional bifurcation that can occur in a system when a critical point changes its stability due to a change in a system parameter. This bifurcation is characterized by the occurrence or disappearance of a limit cycle. In our analysis, we will consider parameter N as the bifurcation parameter while keeping the other parameters constant.\\
The theorem stated below will give the necessary and sufficient condition for the occurrence of the Hopf bifurcation at the interior equilibrium $(\Bar{u},\Bar{v})$.Before proceeding to our main result, for the sake of simplicity, we will introduce new variables into our system. Which are\\
 $ v_{1}=Kv,u_{1}=\frac{mu}{q}$ and $\frac{d\Bar{t}}{dt}=\frac{m}{1+Kv}$.\\    
  By putting these variables in the equation \eqref{main_model_1},our system will transform into ,\\
  \begin{align}
     \begin{cases}
      \frac{du_{1}}{d\Bar{t}}=&a_0u_{1}^N-(a_1u_1^2+a_2u_1+a_3u_1v_1)(1+v_1)\\
       \frac{dv_{1}}{d\Bar{t}}=&v_{1}(u_{1}-1)(1+v_{1})
       \end{cases}     
  \end{align}
 where,$a_0=\frac{r_0m^{N-2}}{q^{N-1}},a_1=\frac{a}{q},a_2=\frac{d}{m},a_3=\frac{p}{km}$ and all are positive.\\
Now that we're familiar with the notations and concepts introduced thus far, we'll proceed with dropping the bar notations. Therefore, our new system will be
   \begin{align}\label{normal form_1}
     \begin{cases}
       \frac{du}{dt}=&a_0u^N-(a_1u^2+a_2u+a_3uv)(1+v)\\
       \frac{dv}{dt}=&v(u-1)(1+v)
       \end{cases}
  \end{align}
  and the equilibrium points of the system are $(0,0),(u',0),(u'',0) $ and $(1,\Bar{v})$ where $u'$ or $u''$ will satisfy the condition $a_0u^N-1-a_1u-a_2=0$ and 
      $\Bar{v}=\frac{-(a_3+a_2+a_1)+\sqrt{[a_3-(a_1+a_2)]^2+4a_0a_3}}{2a_3}$
  Now, jacobian of the above system is given by,
  \begin{align*}
  J'_2=\begin{bmatrix}
  Na_0u^{N-1}-(1+v)(2a_1u+a_2+a_3v) & -a_3u-(a_1u^2+a_2u)-2a_3uv\\
  v(1+v) & 0
  \end{bmatrix}
  \end{align*}
  and at $(1,\Bar{v})$ the jacobian will be 
  \begin{align}\label{new_normal_2}
  J''_2=\begin{bmatrix}
Na_0-(1+\Bar{v})(2a_1+a_2+a_3\Bar{v}) & -a_3-(a_1+a_2)-2a_3\Bar{v}\\
  \Bar{v}(1+\Bar{v}) & 0
  \end{bmatrix} 
  =\begin{bmatrix}
  J'_{11}&J'_{12}\\
  J'_{21}&J'_{22}
  \end{bmatrix}
  \end{align}
 Hence,
 \begin{align*}
     trace(J''_2)=&Na_0-(1+\Bar{v})(2a_1+a_2+a_3\Bar{v})\\
     det(J''_2)=&(a_3+a_1+a_2+2a_3\Bar{v})(\Bar{v}(1+\Bar{v}))>0
\end{align*}
Hence, the characteristic equation of the system will be,
\begin{align*}
    (\omega_H)^2-trace(J''_2)\omega_H+det(J''_2)=0\\
    \implies \omega_H=\frac{trace(J''_2)\pm\sqrt{[trace(J''_2)]^2-4det(J''_2)}}{2}.
\end{align*}
dynamics in the presence of varying degrees of partial cooperation and different levels of predation fear. Drawing
inspiration from the methodology of Morozov et al. (Adamson and Morozov, 2013), we explored how variations in different forms of the functional response term influence the system’s behavior starting with type-II functional response.
Before that, we investigate system dynamics with both one and bi-parameter bifurcation analysis for type-II functional
response separately. It shows a topologically similar outcome to that for the type-I functional response when taking
\begin{theorem}
 When interior equilibrium exists, then Hopf bifurcation will occur at $(\Bar{u},\Bar{v})$ if and only if $N=1+\frac{a_1}{a_1+a_2+a_3\Bar{v}}$.   
\end{theorem}
\begin{proof}
From the Hopf bifurcation theorem we have that  Hopf bifurcation will occour for a parameter value $N=N_H$ if \\
\begin{align}
    Rel[\omega_H](N_H)=&0\nonumber \textit{ \& }
    trace[J''_2](N_H)= N_Ha_0-(1+\Bar{v})(2a_1+a_2+a_3\Bar{v})=0\nonumber\\
   N_H=&1+\frac{a_1}{a_1+a_2+a_3\Bar{v}}\label{condition of hopf bifurcation}
\end{align}
  When N=$N_H$ it is evident that both the eigenvalues are purely imaginary and also$Img[\omega_H](N_H)= \lambda_H(say)\\
  \sqrt{(1+2\Bar{v})(a_0-a_1)+a_1\Bar{v}^2}>0$. Now we are going to verify the transversality condition, which shows that the eigenvalues of our system will pass through the imaginary axis with a non-zero velocity. Hence we take a point N from any neighborhood of $N_H$,
$\bigg(\frac{dRel[\omega_H]}{dN}\bigg)_{N=N_H}=a_0\neq0.$
\end{proof}
This confirms the existence of Hopf bifurcation.
\subsection{Direction of Hopf Bifurcation}
A Taylor series expansion up to the third-order of the system \ref{normal form_1} about the equilibrium point $z=(z_1,z_2)=(0,0)$ yields the following form: 
\begin{equation*}
    \dot{Z} = J_{E^*} Z + H(Z), \text{ where } Z = 
\begin{pmatrix}
z_1 \\
z_2
\end{pmatrix},
\quad
H(Z) = 
\begin{pmatrix}
H_1 \\
H_2
\end{pmatrix}.
\end{equation*}
The nonlinear terms are given by:
\[
H_1 = c_{20} z_1^2 + c_{11} z_1 z_2 + c_{02} z_2^2 + c_{30} z_1^3 + c_{21} z_1^2 z_2 + c_{12} z_1 z_2^2 + c_{03} z_2^3,
\]
\[
H_2 = d_{20} z_1^2 + d_{11} z_1 z_2 + d_{02} z_2^2 + d_{30} z_1^3 + d_{21} z_1^2 z_2 + d_{12} z_1 z_2^2 + d_{03} z_2^3.
\]
Here,
\begin{align*}
c_{20} &= N(N-1)a_0 - 2a_1(1+v), &d_{20} = 0; \\
c_{11} &= -a_3 (1+v) - (2a_1 + a_2 + a_3 v), &d_{11} = 1; \\
c_{02} &= -2a_3, &d_{02} = -1; \\
c_{30} &= N(N-1)(N-2)a_0, &d_{30} = 0; \\
c_{21} &= -2a_1, &d_{21} = 0;\\
c_{12} &= -2a_3, &d_{12} = 1;\\
c_{03} &= 0, &d_{03} = 0.
\end{align*}
Using the transformation \( Z = SY \), where 
$
S = \begin{pmatrix} c_{01} & 0 \\ -c_{10} & -\lambda_H \end{pmatrix},
$
The system becomes:
\begin{equation*}
    \dot{Y} = S^{-1} J_{E^*} S Y + S^{-1} H(SY).
\end{equation*}
Which can be written as:
\[
\begin{pmatrix} \dot{y}_1 \\ \dot{y}_2 \end{pmatrix} = \begin{pmatrix} 0 & -\lambda_H \\ \lambda_H & 0 \end{pmatrix} \begin{pmatrix} y_1 \\ y_2 \end{pmatrix} + \begin{pmatrix} Q_1(y_1, y_2; N = N_H) \\ Q_2(y_1, y_2; N = N_H) \end{pmatrix},
\]
where \( Q_1 \) and \( Q_2 \) are nonlinear functions of \( y_1 \) and \( y_2 \), given by:
\[
Q_1(y_1, y_2; N = N_H) = \frac{1}{c_{01}} H_1,
\hspace{1cm} \& \hspace{1cm}
Q_2(y_1, y_2; N = N_H) = \frac{-1}{\lambda_H c_{01}} \left( c_{10} H_1 + c_{01} H_2 \right).
\]
To calculate the first Lyapunov coefficient \( l_1 \), we use the following expression:
\[
l_1 = \frac{1}{16} \left( Q_1^{111} + Q_1^{122} + Q_2^{112} + Q_2^{222} \right) + \frac{1}{16\lambda_H} \left( Q_1^{12}(Q_1^{11} + Q_1^{22}) - Q_2^{12}(Q_2^{11} + Q_2^{22}) - Q_1^{11} Q_2^{11} + Q_1^{22} Q_2^{22} \right),
\]
where the second and third derivatives of \( Q_1 \) and \( Q_2 \) are evaluated at \( (y_1, y_2; N) = (0, 0; N_H) \):
$
Q_k^{ij} = \frac{\partial^2 Q_k}{\partial y_i \partial y_j} \bigg|_{(y_1, y_2; N) = (0, 0; N_H)}
$ and
$
Q_k^{ijl} = \frac{\partial^3 Q_k}{\partial y_i \partial y_j \partial y_l} \bigg|_{(y_1, y_2; N) = (0, 0; N_H)},
$
where \( i, j, k, l \in \{1, 2\} \). By substituting the expressions of the partial derivatives into the expression for \( l_1 \), we obtain the following result (for detailed calculations, see Supplementary \eqref{lyapunov}).
\begin{align*}
l_1 &= \frac{1}{16} \left( Q_1^{111} + Q_1^{122} + Q_2^{112} + Q_2^{222} \right) \\
&\quad + \frac{1}{16 \lambda_H} \left( Q_1^{12}(Q_1^{11} + Q_1^{22}) - Q_2^{12}(Q_2^{11} + Q_2^{22}) - Q_1^{11} Q_2^{11} + Q_1^{22} Q_2^{22} \right) \\
&= \frac{1}{16} \left( 0 + 0 + 0 + 0 \right) \\
&\quad + \frac{1}{16 \lambda_H} \left[ \lambda_H \left( 2 c_{02} c_{10} - c_{11} c_{01} \right) \left( c_{20} c_{01}^2 - c_{11} c_{01} c_{10} + c_{02} c_{10}^2 + \lambda_H^2 c_{02} \right) \right. \\
&\quad - \left. \lambda_H \left( 2 d_{02} c_{10} - d_{11} c_{01} \right) \left( d_{20} c_{01}^2 - d_{11} c_{01} c_{10} + d_{02} c_{10}^2 + \lambda_H^2 d_{02} \right) \right. \\
&\quad - \left. \left( c_{20} c_{01}^2 - c_{11} c_{01} c_{10} + c_{02} c_{10}^2 \right) \left( d_{20} c_{01}^2 - d_{11} c_{01} c_{10} + d_{02} c_{10}^2 \right) \right. \\
&\quad + \left. \left( \lambda_H^2 c_{02} \right) \left( \lambda_H^2 d_{02} \right) \right]\\
&= \frac{1}{16 \lambda_H} \left[ \left[N(N-1)a_0 - 2a_1(1+v)\right] (-a_3 - a_1 - a_2 - 2a_3 v)^2 \lambda_H^2 - 2a_3 \lambda_H^4 - \lambda_H^3 \left[a_3 + a_1 + a_2 + 2a_3 v\right] + 2a_3 \lambda_H^4 \right] \\
&= \frac{1}{16 \lambda_H} \left[ \left[N(N-1)a_0 - 2a_1(1+v)\right] (-a_3 - a_1 - a_2 - 2a_3 v)^2 \lambda_H^2 - \lambda_H^3 \left[a_3 + a_1 + a_2 + 2a_3 v\right] \right. \\
&\quad \left. + \left(2a_3 \lambda_H^4 - 2a_3 \lambda_H^4\right) \right] \\
&= \frac{1}{16 \lambda_H} \left[ \left[N(N-1)a_0 - 2a_1(1+v)\right] (-a_3 - a_1 - a_2 - 2a_3 v)^2 \lambda_H^2 - \lambda_H^3 \left[a_3 + a_1 + a_2 + 2a_3 v\right] \right]
\end{align*}
\subsubsection*{Observation:}
From the above expression, it is evident that when \(N = N_H, v=v^*\), both \(N a_0 - (2a_1 + a_2 + a_3 v^*)(1 + v^*) = 0\) and \(a_0 - (a_1 + a_2 + a_3 v^*)(1 + v^*) = 0\) hold simultaneously. Consequently, we obtain \(\left[N(N-1) a_0 - 2a_1(1+v)\right] < 0\), which implies \(l_1 < 0\). This indicates that after undergoing a Hopf bifurcation, the system will generate a stable limit cycle, persisting until it experiences the next bifurcation.
\subsection{Heteroclinic Bifurcation}\label{heteroclinic_curve_sum}
To prove the heteroclinic bifurcation, our main objective is to demonstrate the existence of heteroclinic curves. We have already shown that after a supercritical Hopf bifurcation, the system's oscillation amplitude increases without changing direction. Therefore, if a heteroclinic orbit exists, the expanding limit cycle will inevitably collide with this orbit, leading to a heteroclinic bifurcation.
\subsection*{Existence of a Heteroclinic orbit}\label{heteroclinic_curve_sum}
\begin{theorem}
If \(\alpha < \frac{m}{q} < \beta\), a heteroclinic orbit will exist in the system, acting as a separatrix between the stable manifold of the trivial equilibrium and that of the interior equilibrium. This orbit consists of two distinct trajectories connecting the saddle points \((\alpha, 0)\) and \((\beta, 0)\), and it may also arise when the stable limit cycle functions as an interior attractor.%The separatrix is a critical feature as it divides the phase space and reveals the distinct dynamical behaviors and attractors of the system.
\end{theorem}
 \begin{figure}[h!]
    \begin{center}
   \includegraphics[height=2.5in,width=5in]{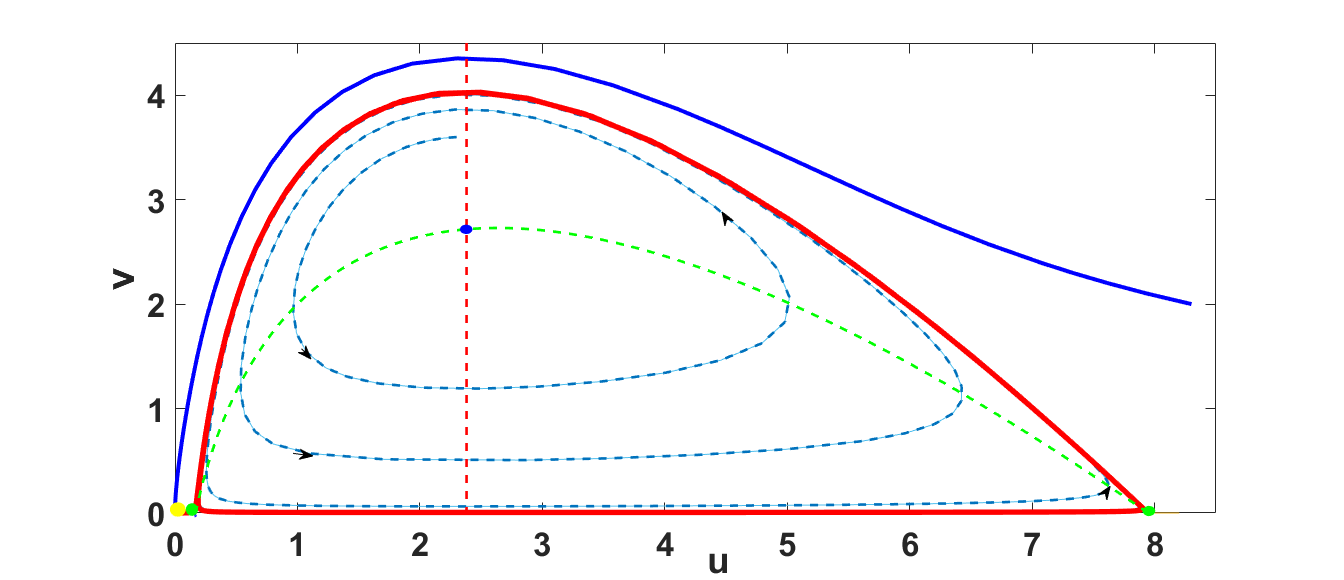}
    \end{center}
    \caption{The red closed orbit connecting two saddle axial equilibria (small green bid) is a heteroclinic curve that acts like a separatrix. Flow lines from the interior approach the orbit, while blue and violet flows originating outside converge to the origin (small yellow bid). The green dotted curve and the red dotted line represent prey and predator nullclines respectively, intersecting at the unique interior equilibrium point (blue bid), which is unstable in this scenario. Parameter values used for this analysis are as follows:\(r = 0.6\), \(N = 1.52712655\), \(K = 0\), \(a = 0.2\), \(p = 0.1\), \(d = 0.2\), \(c = 0.42\), and \(m = 0.1\) }
    \label{heteroclinic_curve}
    \end{figure}
\begin{proof}
Proving the existence of a heteroclinic orbit typically does not follow a standard approach and presents certain difficulties, necessitating the development of a suitable method for confirmation. In this section, we establish the existence of the heteroclinic orbit using the method proposed by \citep{bao2011new}. As this method is not universally applicable, we introduce specific assumptions and transformations to simplify the system into a consistent form. In our case, we focus on demonstrating the existence of the orbit rather than deriving the exact solution or trajectory. To simplify certain computational challenges, we examine the system without considering the fear effect and assign specific parameter values as 1 where necessary.

\[
\frac{dx}{dt} = x^N - ax^2 - dx - xy,
\]
\[
\frac{dy}{dt} = xy - ky.
\]
Now, we want to take a transformation, say \(xdt = d\tau'\), then:
\[
\frac{dx}{d\tau'} = x^{N-1} - ax - d - y,
\]
\[
\frac{dy}{d\tau'} = \frac{1}{x} (xy - ky).
\]

In our study, we seek a heteroclinic curve joining the saddle points \((\alpha, 0)\) and \((\beta, 0)\) of the system. These points satisfy \(\alpha^{N-2} > \frac{a}{(N-1)}\) and \(\frac{a}{(N-1)} > \beta^{N-2}\). Consequently, any \(k\) on this curve must satisfy \(\alpha < k < \beta\).\\

\subsubsection{ Study around the fixed point \textbf{$(\alpha,0)$}}\label{heteroclinic_alpha}
To analyze the dynamical behavior of the trajectories of the heteroclinic orbit around the saddle point \((\alpha, 0)\), we apply the transformation \(x = \alpha + u\) and \(y = v\), where \(\alpha\) satisfies \(\alpha^N - a\alpha^2 - d\alpha = 0\). We assume \(u\) to be very close to zero for this analysis. Thus, the reduced system is given by:
%\[
%\frac{du}{d\tau} = \alpha^{N-1}  (1 + \frac{u}{\alpha})^{N-1} - a\alpha (1 + \frac{u}{\alpha}) - d - v,
%\]
%\[
%\frac{dv}{d\tau} = \frac{1}{\alpha+u} (u + \alpha - k)v.
%\]

%Now, if we simplify the above system:
\[
\frac{du}{d\tau'} = \alpha^{N-1} \left( 1 + (N-1)\frac{u}{\alpha} + \frac{(N-1)(N-2)}{2!} \frac{u^2}{\alpha^2} + \dots \right) - a\alpha - au - d - v,
\]
\[
\frac{dv}{d\tau'} = \frac{1}{\alpha} \left( 1 - \frac{u}{\alpha} + \dots \right) (u + \alpha - k)v.
\]
Simplifying the system by excluding higher-order terms of degree two and above, the reduced system in the more familiar coordinate system is given by:
\begin{align}\label{reduced_heteroclinic}
\frac{du}{dt} &= \left( (N-1)\alpha^{N-2} - a \right) u - v,\nonumber\\
\frac{dv}{dt} &= \frac{k}{\alpha} uv - \frac{k - \alpha}{\alpha}v. 
\end{align}
It is important to note that after this transformation, the limiting condition \(\tau' \Rightarrow +\infty\) remains equivalent to \(t \Rightarrow +\infty\).
To analyze the dynamics of heteroclinic orbits near $(\alpha,0)$, it suffices to examine the trajectory $(u(t),v(t))$ within the first quadrant, ensuring biological plausibility.
\item{Case: for t > 0}\label{Case: for $t>0$}
\begin{itemize}
As in all the discussion, based on the assumption that \( t > 0 \), hence after taking the transformation
\begin{align}\label{t_to_tau}
    t = -\frac{1}{T_1} \ln(\tau)
\end{align}
the possible series expansion of $u(t)$ and $v(t)$ become,
\begin{align}\label{sum_gen_heteroclinic}
    0\leq u(\tau)= 0+\sum_{i=1}^{\infty}a_i\tau^i\\
    0\leq v(\tau)= 0+\sum_{i=1}^{\infty}b_i\tau^i.
\end{align}
%becomes:
%\[
%-T_1 \tau \dot{u}(\tau) = \left( (N-1)\alpha^{N-2} - a \right) u - v,
%\]
%\[
%-T_1 \tau \dot{v}(\tau) = \frac{k}{\alpha} uv - \frac{k-\alpha}{\alpha} v.
%\]
Substituting the values into the system \eqref{reduced_heteroclinic},
\begin{align*}
  -T_1 a_1 \tau - 2T_1 a_2 \tau^2 - 3T_1 a_3 \tau^3 - \dots &= \Delta \sum_{i=1}^{\infty} a_i \tau^i - \sum_{i=1}^{\infty} b_i \tau^i,\\
-T_1 b_1 \tau - 2T_1 b_2 \tau^2 - 3T_1 b_3 \tau^3 - \dots &= \left(\frac{k}{\alpha} \sum_{i=1}^{\infty} a_i \tau^i \sum_{i=1}^{\infty} b_i \tau^i \right)- \frac{k-\alpha}{\alpha} \sum_{i=1}^{\infty} b_i \tau^i.  
\end{align*}

Here, $\Delta = (N-1)\alpha^{N-2} - a $. Comparing the coefficients of \(\tau^i\) in the above equations,
\[
-T_1 a_1 - \Delta a_1 + b_1 = 0 \hspace{0.5cm}\& \hspace{0.5cm}-T_1 b_1 + \frac{k-\alpha}{\alpha} b_1 = 0.
\]
This means that, 
\[
\begin{pmatrix}
-T_1 - \Delta & 1 \\
0 & -T_1 + \frac{k-\alpha}{\alpha}
\end{pmatrix}
\begin{pmatrix}
a_1 \\
b_1
\end{pmatrix} =
\begin{pmatrix}
0 \\
0
\end{pmatrix}.
\]
For the nontrivial solution, we can obtain,
\[
T_1 = -\Delta \quad \text{or} \quad T_1 = \frac{k}{\alpha} - 1 > 0 \quad \text{(as \(k > \alpha\))}.
\]
Also, \(\Delta = (N-1)\alpha^{N-2}\) - \( a  > 0 \) by (4.4) means that $T_1=-\Delta<0$. Hence, the possible $T_1=\frac{k-\alpha}{\alpha}$.\\
Now comparing the rest of the coefficients of $\tau^i$, we have
\begin{align*}
   a_1&=\theta_1, b_1=\left(T_1+\Delta\right) \theta_1 ,[\text{where } \theta_1 \text{is arbitrary real number}];\\
   a_2&=\frac{\frac{k}{\alpha}}{-T_1} \frac{T_1+\Delta}{(2 T_1+\Delta)} \theta_1^2  \Rightarrow \left|a_2\right| \leq \left|\frac{k}{\alpha T_1}\theta_1^2\right|, b_2 = \frac{k}{\alpha T_1}(T_1+\Delta)\theta_1^2 \Rightarrow \left|b_2\right|\leq \left|\frac{k}{\alpha T_1}(T_1+\Delta)\theta_1^2\right|;\\
a_3&=\frac{1}{2}\left(\frac{k}{\alpha T_1}\right)^2(T_1+\Delta) \dot{\theta}^3\left[1+\frac{T_1+\Delta}{2 T_1+\Delta}\right] \frac{1}{3 T_1+\Delta} \Rightarrow \left|a_3\right| \leq \left|\left(\frac{k}{\alpha T_1}\right)^2 \theta_1^3\right| ,\\
b_3 %=\frac{k}{\alpha} \frac{\left(a_1 b_2+b_1 a_2\right)}{-2 T} 
%& =\frac{\frac{k}{\alpha}}{-2 T}\left[\frac{k / \alpha}{-T}(T_1+\Delta) \theta_1^3+\frac{\frac{k}{\alpha}}{-T} \frac{(T_1+\Delta)^2}{(2 T_1+\Delta)} \theta_1^3\right] 
&=\frac{1}{2}\left(\frac{k}{\alpha T_1}\right)^2(T_1+\Delta) \theta_1^3\left[1+\frac{T_1+\Delta}{2 T_1+\Delta}\right] \Rightarrow \left|b_3\right| \leq \left|\left(\frac{k}{\alpha T_1}\right)^2(T_1+\Delta) \theta_1^3\right|;\\
a_4&=\frac{\frac{k}{\alpha}}{-6 T_1}\left(\frac{k}{\alpha T_1}\right)^2(T_1+\Delta) \theta_1^4\left[\left[1+\frac{T_1+\Delta}{2T_1+\Delta}\right]\left[1+\frac{T_1+\Delta}{3 T_1+\Delta}\right]+2\left[\frac{T_1+\Delta}{2 T_1+d}\right]\right]\\
& \hspace{6cm} \left[\frac{1}{4 T_1+\Delta}\right]\Rightarrow \left|a_4\right| \leq \left|\left(\frac{k}{\alpha T_1}\right)^3 \theta_1^4\right|,\\
b_4&=\frac{\frac{k}{\alpha}}{-6 T_1}\left(\frac{k}{\alpha T_1}\right)^2(T_1+\Delta) \theta_1^4\left[\left[1+\frac{T_1+\Delta}{2 T_1+\Delta}\right]\left[1+\frac{T_1+\Delta}{3 T_1+\Delta}\right]+2\left[\frac{T_1+\Delta}{2 T_1+d}\right]\right]\\
& \hspace{6cm}\Rightarrow \left|b_4\right| \leq \left|\left[\left(\frac{k}{\alpha T_1}\right)^3 (T_1+\Delta) \theta_1^4\right]\right|;\dots.
\end{align*}
Hence we can write the general terms as
\begin{align}\label{general_term_heteroclinic}
   \left|a_n\right| \leq \left|\left(\frac{k}{\alpha T_1}\right)^{n-1} \theta_1^n\right| \textit{ and }
   \left|b_n\right| \leq \left|\left[\left(\frac{k}{\alpha T_1}\right)^{n-1} (T_1+\Delta) \theta_1^n\right]\right|. 
\end{align}
Now if we put the results \ref{general_term_heteroclinic} in the series expansion \ref{sum_gen_heteroclinic}, we can have\\
\begin{align}\label{ineq_series_final_1}
    0 \leq u(\tau) \leq \sum_{i=1}^{\infty} a_{k}\tau^{k} \leq \sum_{i=1}^{\infty} |a_{i}|\tau^{i} = \sum_{n=1}^{\infty}  \left|\left(\frac{k}{\alpha T_1}\right)^{n-1} \theta_1^n\right|\tau^{n}
\end{align} 
\begin{align}\label{ineq_series_final_2}
    0 \leq v(\tau) \leq \sum_{i=1}^{\infty} b_{i}\tau^{i} \leq \sum_{i=1}^{\infty} |b_{i}|\tau^{i} = \sum_{n=1}^{\infty}  \left|\left(\frac{k}{\alpha T_1}\right)^{n-1} (T_1+\Delta)\theta_1^n\right|\tau^{n}
\end{align} 
Hence $R_1$ be the radius of convergence of the power series, then $R_1 = \lim_{n\to\infty} \frac{\left(k/\alpha T_1\right)^{1/n}}{\frac{k}{\alpha T_1}\theta_1}= \frac{1}{\frac{k}{\alpha T_1}\theta_1} $\\
Now transforming result \ref{ineq_series_final_1} in the variable $t$ with $t\to +\infty$ then we have,
\begin{align*}
    0 \leq u(\tau) \leq  \sum_{n=1}^{\infty}  \left|\left(\frac{k}{\alpha T_1}\right)^{n-1} \theta_1^n\right|\tau^{n} = \frac{\theta_1 e^{-\frac{k-\alpha}{\alpha}t}}{1-\frac{k}{\alpha T_1}\theta_1 e^{-\frac{k-\alpha}{\alpha}t}}\to 0.
\end{align*}
Similarly, from \ref{ineq_series_final_2} we can have,
\begin{align*}
    0 \leq v(\tau) \leq  \sum_{n=1}^{\infty}  \left|\left(\frac{k}{\alpha T_1}\right)^{n-1} \theta_1^n\right|\tau^{n} = \frac{\theta_1 e^{-\frac{k-\alpha}{\alpha}t}}{1-\frac{k}{\alpha T_1}\theta_1 (T_1+\Delta)e^{-\frac{k-\alpha}{\alpha}t}}\to 0.
\end{align*}
So, we can conclude that,$(x(t),y(t)) \to (\alpha,0)$ as $t \to +\infty$.\\ 
\item{Case: for $t<0$: }\label{Case: for $t<0$}
Now, in our case, if we try to study the consequences for \(t < 0\), we have to choose
$t = \frac{1}{T_2} \ln \tau$
for the system
\begin{align}
    \dot{u}(\tau) &= \left[ (N-1)\alpha^{N-2} - a \right] u + \frac{(N-1)(N-2)}{2} u^2 \alpha^{N-3} - v\nonumber,\\
    \dot{v}(\tau) &= \frac{k}{\alpha} uv - \frac{k - \alpha}{\alpha} v.
\end{align}
Let us assume that $\Delta = \left[ (N-1)\alpha^{N-2} - a \right]$ and $\Phi = \frac{(N-1)(N-2)}{2}\alpha^{N-3}$ and proceeding as before, we conclude that we now have
\[
\begin{pmatrix}
T_2 - \Delta & 1 \\
0 & T_2 + \frac{k - \alpha}{\alpha}
\end{pmatrix}
\begin{pmatrix}
a_1' \\
b_1'
\end{pmatrix}
= 
\begin{pmatrix}
0 \\
0
\end{pmatrix}.
\]
Hence, \(T_2 = \Delta\) and comparing the coefficients of $\tau^1$ we can say $b_1= 0$ and we assume , $a_1= \xi$ where $\xi$ is an arbitrary constant. Now, comparing the coefficients of \(\tau^i\) on both sides, we conclude that:
\begin{align*}
 T_2 a_2 &= \frac{\Phi}{\Delta} \xi^2 \Rightarrow a_2 = \frac{\Phi}{\Delta} \xi^2,\\
2 T_2 a_2 &= \Delta a_2 - b_2 + \Phi a_2^2 \Rightarrow 2 T_2 b_2 = \frac{k}{\alpha} a_1 b_1 + \frac{k-\alpha}{\alpha} b_2 \Rightarrow \quad b_2 = 0;\\  
3 T_2 a_3 &= \Delta a_3 - b_3 + 2\Phi a_2 a_3 \Rightarrow a_3 = \left(\frac{\Phi}{\Delta} \right)^2 \xi^3 \textit{ and } b_3=0;\\
3 \Delta a_4 &= \Phi (a_1 a_3 + a_3 a_1 + a_2^2) \Rightarrow a_4 = \left( \frac{\Phi}{\Delta} \right)^3 \xi^4 \textit{ and } b_4=0;\\
\dots
\end{align*}
Hence following similar steps we can now conclude that,
\begin{align*}
 a_n &= \left( \frac{\Phi}{\Delta} \right)^{n-1} \xi^n \textit{ and } b_n=0 ;\textit{ for all } n \in \mathbb{N}.\\
    \text{ Thus we can say that, } u(\tau) &= \sum_{n=1}^\infty \left( \frac{\Phi}{\Delta} \right)^{n-1} \xi^n \tau^n \textit{ and } v(\tau) = 0.\\
    \text{ That means, } u(t) &= \frac{\xi e^{\Delta t}}{(1 - \frac{\Phi}{\Delta} \xi e^{\Delta t})} \to 0 \text{ as } t \to -\infty \textit{ and } v(t) = 0.
\end{align*}
Hence, around the saddle-fixed point \((\alpha, 0)\), our study confirms that one branch of the heteroclinic orbit, which lies over the \(x\)-axis, will converge to \((\alpha, 0)\) as \(t \to -\infty\). On the other hand, the other curvilinear trajectory, with \(x(t) > 0, y(t) > 0\) when \(t > 0\), will converge to \((\alpha, 0)\) as \(t \to +\infty\).
\end{itemize}
\subsubsection{ Study around the fixed point \textbf{$(\beta,0)$}}
To analyze the behavior of the trajectory near the saddle point \((\beta, 0)\), we perform a variable transformation \(x = u + \beta\) and employ a methodology analogous to that outlined in Section~\ref{heteroclinic_alpha}. This allows us to draw parallel conclusions regarding the dynamical structure around \((\beta, 0)\).
\begin{itemize}
\item{Case: for $t > 0$: }%\label{Case: for $t>0_\beta$}
But now for the study of \(t > 0\), we will choose
$t = - \frac{1}{T'_1} \ln \tau$ then putting this in the transformed system we will get $T_1' = (N-1) \beta^{N-2} - a ( by \ref{beta_position})\text{ as } \frac{k - \beta}{\beta} < 0$. Now, following the same steps as deducted in \ref{Case: for $t<0$}, we get:
\begin{align*}
   u(t) = \sum_{n=1}^\infty \left( \frac{\Phi_1}{\Delta_1} \right)^{n-1} \xi^n \tau^n = \frac{\xi e^{-\Delta_1 t}}{(1 - \frac{\Phi_1}{\Delta_1} \xi e^{-\Delta_1 t})} \to 0 \text{ as } t \to +\infty \textit{ and } v(t) = 0.\\
\text{Where, } \Phi_1= \frac{(N-1)(N-2)}{2}\beta^{N-3},\Delta_1= \left[ (N-1)\beta^{N-2} - a \right]. 
\end{align*}
Thus, as \(t \to +\infty\), one branch of the heteroclinic orbit, which lies along the \(x\)-axis, will asymptotically approach the saddle point \((\beta, 0)\), following the trajectory described earlier.
\item{Case: for $t<0$: }%\label{Case: for $t<0_\beta$}
By considering the case for \(t < 0\) and assuming the transformation \(t = \frac{1}{T'_2} \ln \tau\), we obtain \(T'_2 = \frac{\beta - k}{\beta}\). Thus, following the steps outlined in \ref{Case: for $t>0$}, we arrive at:

\[
u(t) = \frac{\theta_1 e^{\frac{\beta - k}{\beta}t}}{1-\frac{k}{\beta T'_2}\theta_1 e^{\frac{\beta - k}{\beta}t}}, \quad v(t) = \frac{\theta_1 e^{\frac{\beta - k}{\beta}t}}{1-\frac{k}{\beta T'_2}\theta_1 (T'_2+\Delta_1)e^{\frac{\beta - k}{\beta}t}}.
\]
The trajectory of the other heteroclinic orbit remains in the positive quadrant and converges to $(\beta,0)$ as $t \to -\infty$. 
\subsection*{Observation:}
Our analysis of the heteroclinic orbit not only tracks the potential paths of the heteroclinic curves near the saddle points but also identifies two distinct trajectories governed by the time variables. One trajectory lies along the x-axis, originating at \((\alpha, 0)\) and converging to \((\beta, 0)\), while the second trajectory, with nonzero ordinates, starts from the opposite saddle point and converges to \((\alpha, 0)\) as \(t \Rightarrow \infty\). These results precisely characterize the directionality and structure of the heteroclinic orbits, providing new insights into our system’s dynamics. Our method also indicates that the presence of two distinct positive eigenvalues suggests the possibility of two separate trajectories, confirming the existence of a heteroclinic orbit. 
\end{itemize}
\end{proof}
\section{Global Stability}
We've demonstrated that when \( N = N_H \), the interior equilibrium becomes unstable due to a Hopf bifurcation. At this point, all equilibrium points except the origin are unstable, yet a stable limit cycle emerges in the system. Upon surpassing a critical threshold value of \( N \), the stability of the limit cycle diminishes via a heteroclinic bifurcation, leaving the origin as the sole attractor in the system.\\

We will now establish the global stability of the origin mathematically using the Bendixon-Dulac stability criterion \citep{perko2013differential}. \\
For this proof, we assume that $N\in(1,2)$ and $x,y>0$
Let, us also assume that, \\
\begin{align*}
    f_1(x,y) & = \frac{r_0 x^N}{1+Ky} - ax^2 - pxy -dx , f_2(x,y)= cpxy - my 
    \textit{ $ \& $ } h(x,y) = x^{-N}.\\
    \textit{Now, } \Delta &= \frac{\delta}{\delta x}( h(x,y)f_1(x,y))+\frac{\delta}{\delta y}( h(x,y)f_2(x,y))\\
    &= \frac{\delta}{\delta x}(\frac{r_0}{1+Ky}-ax^{2-N}-px^{1-N}y-dx^{1-N})+\frac{\delta}{\delta y}(qx^{1-N} - mx^{-N})\\
    &\geq -a(2-N)x^{1-N}-d(1-N)x^{-N}+qx^{1-N}-mx^{-N}\\
    & = [q-a(2-N)]x^{1-N}+[d(N-1)-m]x^N.
\end{align*}
If both conditions \( N > 1 + \frac{m}{d} \) and \( N > 2 - \frac{q}{a} \) are satisfied simultaneously, or if exactly one of these two becomes equality, then \( \Delta > 0 \), indicating the absence of any closed orbits in the system. Consequently, when \( N > N'_0 \), all equilibrium points except the origin become unstable. Thus, when \( N > \max\{N'_0, 1 + \frac{m}{d}, 2 - \frac{q}{a}\} \), the origin becomes globally asymptotically stable.In the following section, we verify this condition numerically for a value of $N$ that exactly satisfies the aforementioned result.
\section{Numerical Validation}
In addition to rigorously validating the mathematical findings elucidated in the preceding section through numerical simulations. We initiate our investigation by considering a type-I functional response, subsequently extending our analysis to encompass the dynamics under the type-II functional response. We conduct our numerical simulations using MATLAB 2016a and MatCont 6 software. For $N=1$, the system assumes the exact form utilized in the work of \citep{wang2016modelling}, where the dynamics of the model have been extensively studied for both type-I and type-II functional responses. Therefore, in our case, we begin our exploration of bifurcation analysis for $1<N<2$.  \\
\begin{figure}[h!]
\centering
 \subfloat[\centering ]
{\includegraphics[height=8.0cm,width=\textwidth]{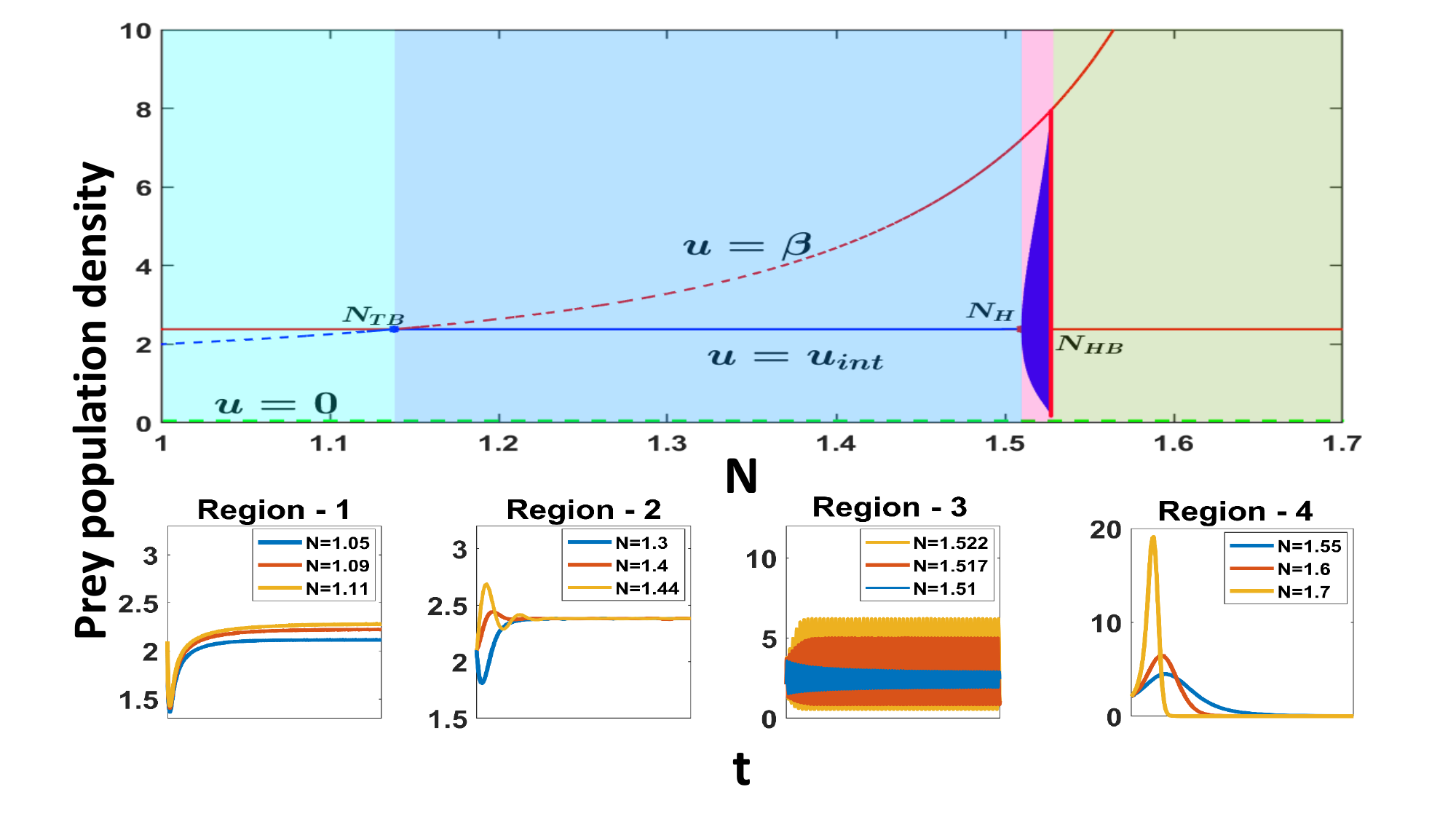}} 
\quad
   \subfloat[\centering] % This situation is very similar to the previous figure except the amplitude of the oscillation increases due to a little increase of the value N=1.22 rest of the parameters remain the same.]
   {\includegraphics[height=6.8cm,width=\textwidth]{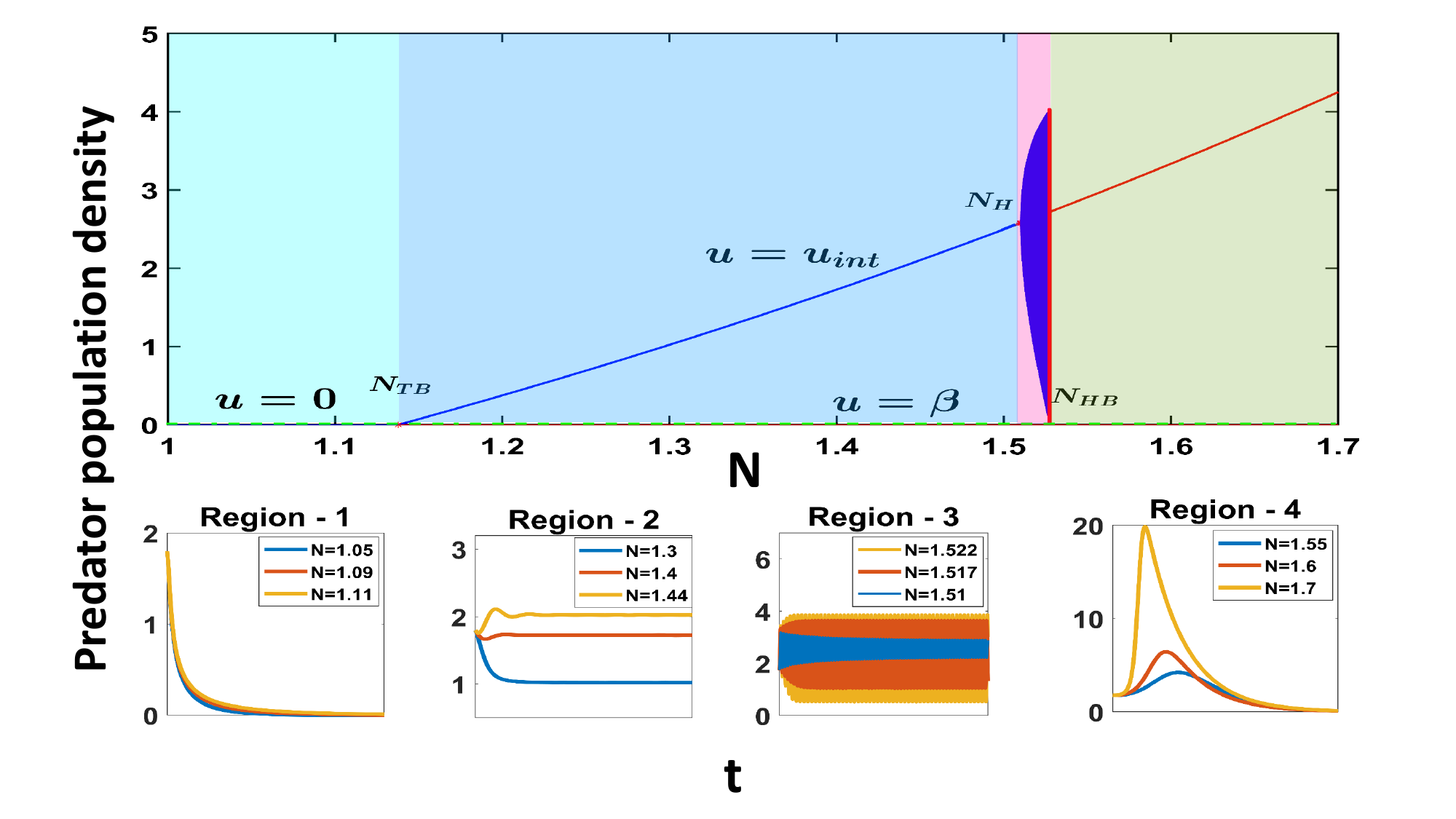}}
   \caption{All the one-parameter bifurcations of our system concerning the parameter $N$, plotted against prey density (see \((a)\)) and predator density (\(see (b)\)) respectively. Supporting region-wise change time series representation for different values of N is demonstrated in the text. In these figures, $u=\beta$ represents the bifurcation curve for right axial equilibrium, and $u=u_{int}$ is for the interior equilibrium. For the trivial equilibrium we use the $u=0$ line, its stable state is represented by a green dashed line. } 
    \label{bifurcation_timeseries}
     \label{timeseries_axial_v_1}
%\caption{Caption for this figure with two images}
%\label{fig:image2}
\end{figure}
 However, we employ a different parameter set only for better representation, although the results exhibit the same dynamical properties. The newly chosen parameter set is as follows:
$r = 0.6, \quad a = 0.2, \quad p = 0.1, \quad c = 0.42, \quad m = 0.1, \quad d = 0.2.$ We have taken $r$ instead of $r_0$ at the time of simulation.
\subsection{Study without fear effect}
In the absence of fear (i.e., \( K = 0 \)), as the parameter \( N \) initially increases, the predator nullcline does not intersect with the prey nullcline \eqref{bifurcation_timeseries}, leaving the interior equilibrium unstable. Meanwhile, the non-saddle predator-free equilibrium \( E_2 = (\beta, 0) \) remains stable up to \( N = 1.1654321 \) (\( N_{TB} \)). Additionally, we notice that the equilibrium density of the prey species increases accordingly with the increasing value of $N$ until it reaches $N_{TB}$. Due to the existence of the extinction equilibrium at a stable state, in this domain of the parameter "N," the system experiences a bi-stability between the predator-free equilibrium and the extinction equilibrium. We use light sky color to cover this region (see $\mathcal{R}_1$(Region-1) in Figure\eqref{bifurcation_timeseries}). At the critical point $N = N_{TB}$, the predator-free equilibrium $E_2 = (\beta,0)$ loses its stability through Transcritical bifurcation, after which the coexistence equilibrium becomes stable. This equilibrium is the only interior equilibrium of the system as this time the prey and predator nullcline intersect at only one point (see Figure\eqref{label fig 1}). Suppose we further increase the parameter value $N$ from $N_{TB}$ up to the value $N = N_H$. In that case, we observe that the interior equilibrium not only maintains its stability but also experiences a simultaneous increase in the predator equilibrium density, as we can see in $\mathcal{R}_2$(Region-2) of Figure\eqref{bifurcation_timeseries}(b). Interestingly this time increase in the value of $N$ can not have any impact on the prey density at this steady state, hence remains unchanged (see, $\mathcal{R}_2$ of Figure\eqref{bifurcation_timeseries}(a) and Figure\eqref{bifurcation_timeseries}(b)). A light blue shade is used to highlight the portion $\mathcal{R}_2$ in Figure\eqref{bifurcation_timeseries} representing the above dynamics. Since the trivial equilibrium still remains stable in this region, the bistability between the coexistence equilibrium and the extinction equilibrium also persists here. \\
So we can conclude that partial cooperation can promote the coexistence of the prey and predator species within the system for any chosen value $N$ in between a critical threshold interval $ (N_{TB}, N_H)$. By losing stability at $N = 1.5099327842$ ($N_H$), the coexistence equilibrium becomes an unstable node and stable oscillation takes place in the system through Hopf-bifurcation. Numerical evaluation of the first Lyapunov coefficient yields $L_1 = -0.003608117$, confirming the supercritical nature of the Hopf-bifurcation. Consequently, a stable limit cycle emerges around the interior equilibrium. Interestingly, as \(N\) increases, the stable limit cycle expands, as seen in the light pink-shaded region. For better understanding, we depict the time series representation of predator and prey, which clearly shows the increasing amplitude of oscillations with higher values of $N$ (in region \(\mathcal{R}_3\)(Region-3), Figure\eqref{bifurcation_timeseries}). However, we have demonstrated both numerically and mathematically the existence of the heteroclinic orbit \ref{he}. This closed orbit occurs by connecting the two axial saddle points \(E_1\) and \(E_2\), around the interior equilibrium point. These orbits serve as the separatrix curve, delineating the basins of attraction of the two attractors present in \(\mathcal{R}_3\), the stable limit cycle, and the trivial equilibrium that remains at the stable state. Hence, in this region, density-dependent system extinction remains a possible scenario.
\begin{figure}[h!]
\begin{center}
\includegraphics[height=4in,width=\textwidth]{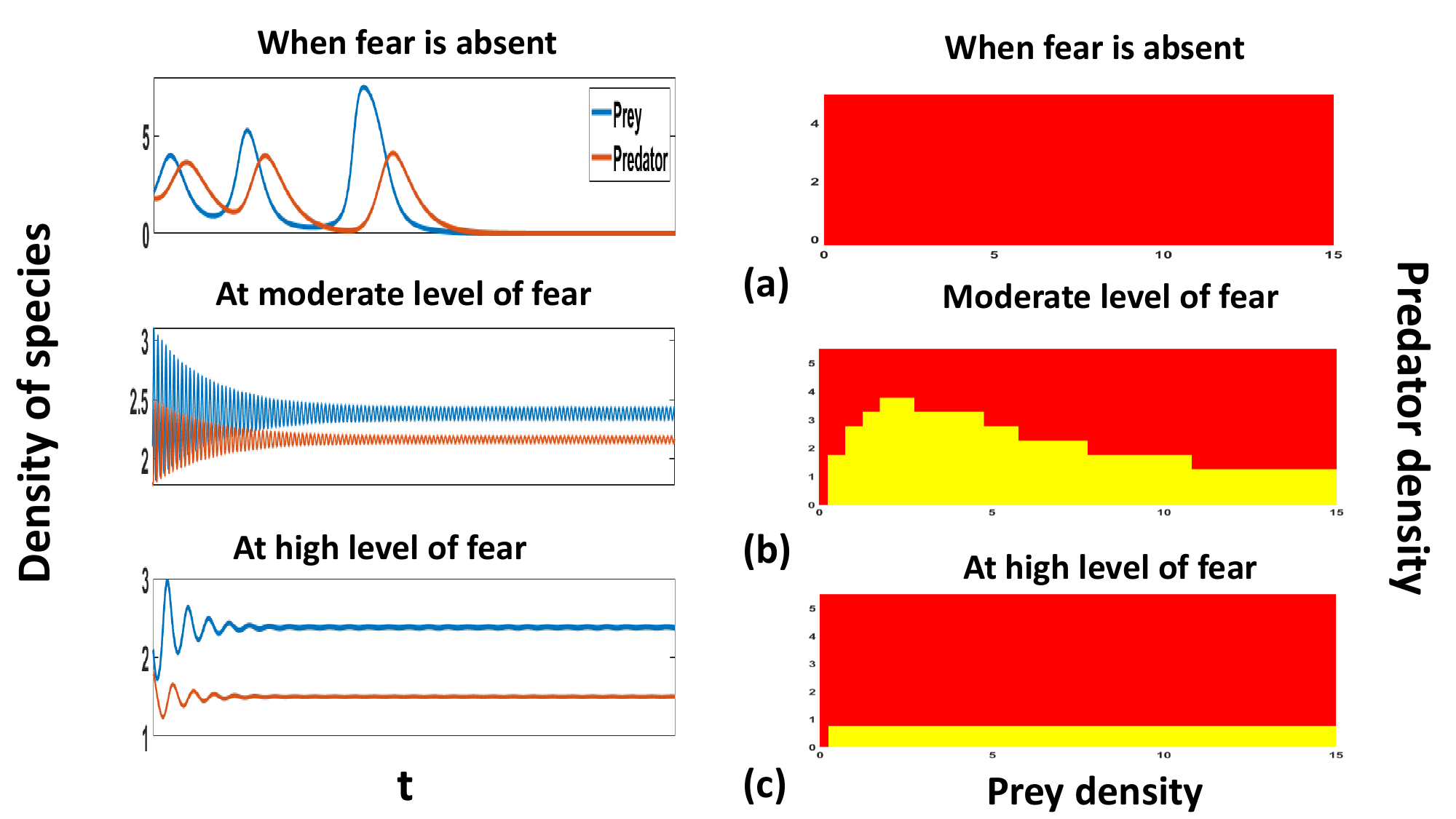}
    \end{center}
    \caption{In this figure, the left-hand side shows the time series trajectories of the prey and predator for different levels of fear, while the right-hand side displays their corresponding basins of attraction. The orange trajectory represents the population density of prey species, while the blue trajectory represents that of the predator. Here, we use the red color to denote the basin of attraction for the extinction equilibrium and the yellow for the corresponding stable state. For this study, we set \( N = 1.53 \) while keeping all other parameters unchanged. We examine the system under different levels of fear, considering three distinct cases: (a) \( K = 0 \), (b) \( K = 0.023 \), and (c) \( K = 0.1 \), with further details outlined in the text.}
    \label{study_impact_of_fear}
    \label{fear timeseries}
    \label{fear basin}
    \end{figure}
    The existence of the heteroclinic orbits also implies the amplitude of the oscillations cannot indefinitely increase with the increment of the parameter value $N$ which is also an infeasible case in a bounded system. However, findings reflect that the increment of the amplitude of the oscillatory trajectories continues until $N$ reaches the value $N_{HB} = 1.527223439$, where the stable limit cycle collides with the heteroclinic orbits and disappears through a heteroclinic bifurcation. We covered the rest of the region of the figures Figure\eqref{bifurcation_timeseries}(a) and Figure\eqref{bifurcation_timeseries}(b) by light olive. This region shows after crossing the value $N_{HB}$, regardless of the value $N$ takes, all trajectories converge to the extinction equilibrium. As observed, the green dashed line in each of the four sectors in Figure \eqref{bifurcation_timeseries} confirms the presence of a stable extinction equilibrium in this region. Additionally, the absence of other stable states in region \(\mathcal{R}_4\)(Region-4) suggests that the origin is now a globally stable equilibrium. Interestingly, the time series representation of this region reflects that as we increase the value of \( N \) in \(\mathcal{R}_4\), prey species face rapid extinction, which is subsequently followed by the extinction of predators.

\subsection{Study with the effect of fear:}
\begin{figure}[h!]
\centering
 \subfloat[\centering]
{\includegraphics[height=5.5cm,width=7.0cm]{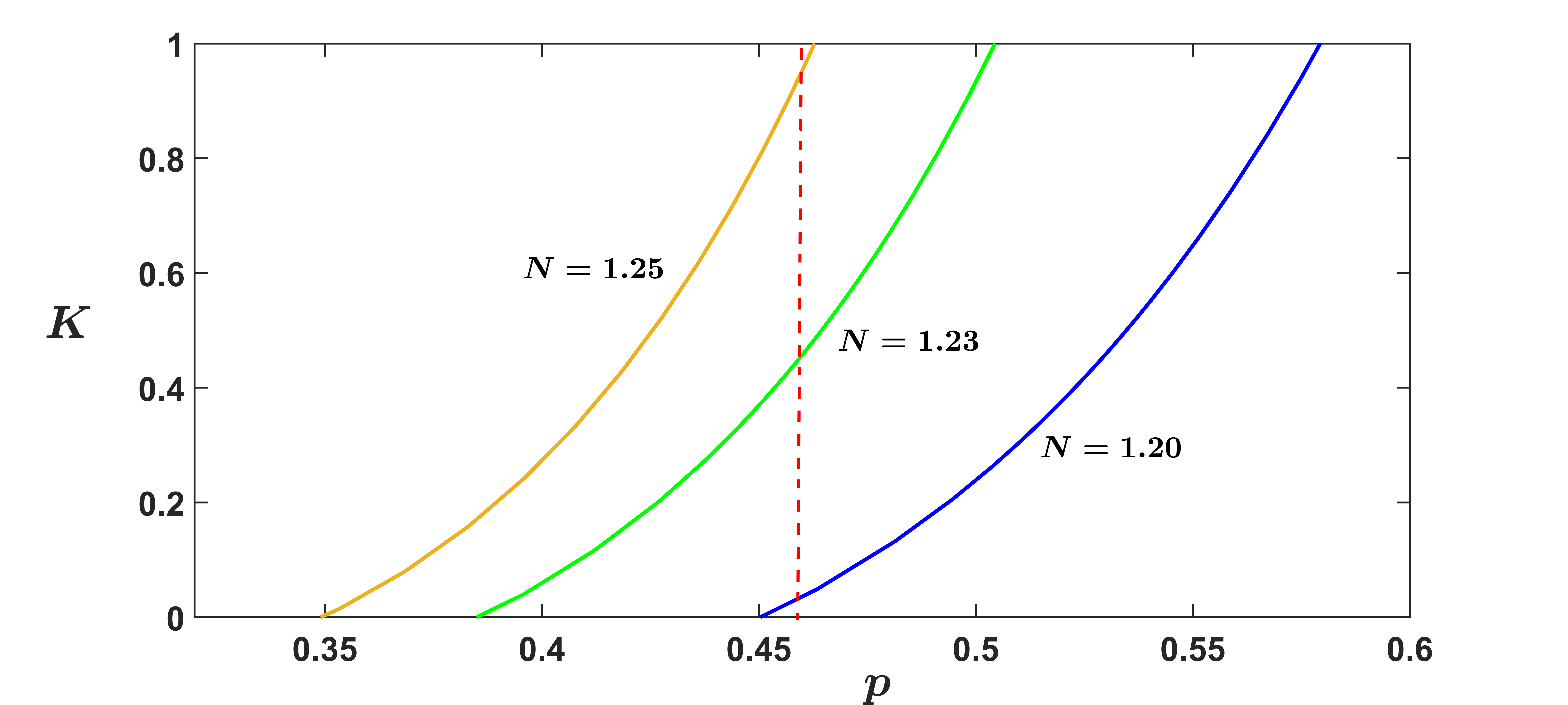}} 
\quad
   \subfloat[\centering]{\includegraphics[height=5.5cm,width=7cm]{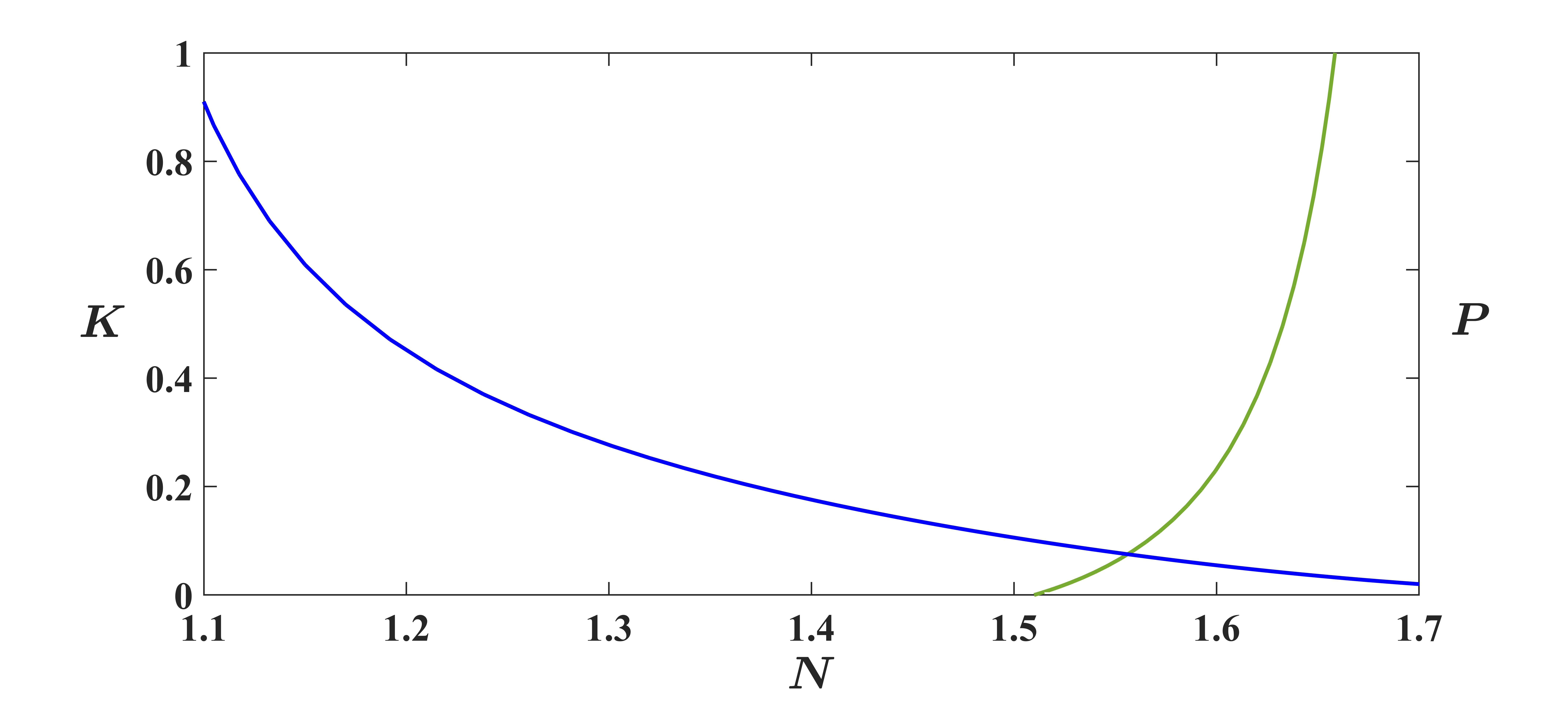}}
   %\quad
   %\subfloat[\centering Now the value of N increased to 1.45 but the rest of the parameters remain the same. Here all the trajectories ultimately approach the trivial equilibrium, making it globally stable.]{\includegraphics[height=4cm,width=7.5cm]{supplymentary_extinction.png}}
    \caption{A two-dimensional bifurcation of the Hopf bifurcation curve is projected across various parameter spaces. Figure \(\ref{label biparameter}(a)\) illustrates the biparameter effect in the \(K-p\) space at three different degrees of partial cooperation. The red dashed line indicates that at a fixed predation rate, progressively larger values of \(K\) are required to achieve a Hopf bifurcation as the degree of partial cooperation increases. In Figure \(\ref{label biparameter}(b)\), the parameter range for \(N\), \(K\), and \(p\) is shown. Here, the area below the blue line represents the parameter combination for stable states, and the above is unstable. For the green curve, the left area represents the parameter combination range for stable states and the right area for unstable states. The blue curve represents the Hopf bifurcation curve for the \(N-p\) plane, while the green curve represents the same for the \(N-K\) plane. The parameters \(K\) and \(p\) are on the same scale and limit, with all other parameters held constant.}%
    \label{fig:example}
     \label{label biparameter}
%\caption{Caption for this figure with two images}
%\label{fig:image2}
\end{figure}
Our next focus is to delve deeper into exploring the impact of fear on the system across different regions of \(N\) as represented in Figure \eqref{bifurcation_timeseries}. Now from the analysis of Figure\eqref{fear basin}, we can infer that fear exerts a stabilizing effect on our system. In Figure \eqref{study_impact_of_fear}(a), where \(K = 0\) and \(N = 1.53\), the origin acts as the sole attractor and exhibits global stability. In contrast, as shown in Figure \eqref{study_impact_of_fear}(b), when \(K\) increases to 0.023, the interior equilibrium remains unstable. However, a stable limit cycle forms around the interior equilibrium, signifying oscillatory dynamics within the system. However, after surpassing a critical fear threshold (\(K > 0.065\)), the interior equilibrium regains its stability, as shown in \eqref{study_impact_of_fear}(c). With further increases in \(K\), the stability of the interior equilibrium remains unchanged, although predator density at the stable equilibrium point continues to decline. Notably, the system cannot revert from stable coexistence to a predator-free stable state, regardless of the increasing \(K\), as the stability of the axial equilibrium is independent of the fear effect, what established mathematically. Comparing the regions of the basin of attraction in Figure \eqref{study_impact_of_fear}(a) and Figure \eqref{study_impact_of_fear}(b), we observe that an initial increase in fear level enlarges the stability region, thereby enhancing the chances of species survival. However, as fear further intensifies, it succeeds in producing an asymptotically stable coexistence state by safeguarding the interior equilibrium from oscillations, but the overall area of the stable region decreases. This showcases a rare yet fascinating duality of the fear effect in our system. In Figure \eqref{label biparameter}, we present a bi-parameter analysis to investigate the interaction between predation fear and cost-associated cooperation. Figure \eqref{label biparameter}(a) explores the dynamics of \( K \) versus \( p \) across varying degrees of partial cooperation, revealing that as predation pressure intensifies, prey species increasingly engage in anti-predator behavior, consistent with established findings \citep{wang2016modelling}. Notably, our results indicate that at a fixed predation rate (where the red dashed line is drawn), the intensity of anti-predator movement escalates as the degree of partial cooperation increases(the red dashed line intersects the corresponding Hopf bifurcation curve at a higher value of $K$), highlighting a novel insight into how cooperative strategies amplify behavioral responses to predation risk. On the other hand, Figure \eqref{label biparameter}(b) demonstrates that at higher levels of partial cooperation, a greater level of fear is required to induce a Hopf bifurcation (represented by the green curve in Figure \eqref{label biparameter}(b)). This suggests that with greater partial cooperation, prey are more inclined to exhibit anti-predator behavior, further highlighting the benefits of cooperation. In contrast, the blue curve in $N-p$ space reveals an opposing scenario: at higher predation rates, even a small degree of partial cooperation can induce a Hopf bifurcation, making the system more susceptible to destabilization at this stage. Furthermore, Wang et al.\citep{wang2016modelling} demonstrated that the dynamics of a prey-predator system in the presence of fear is susceptible to the functional response structure. This motivates us to investigate how various functional response structures influence system dynamics in the presence of varying degrees of partial cooperation and different levels of predation fear. Drawing inspiration from the methodology of Morozov et al. \citep{adamson2013can}, we explored how variations in different forms of the functional response term influence the system's behavior. Before that, we investigate system dynamical analysis for type-II functional response separately. As it shows a topologically similar outcome to that for the type-I functional response when taking $N$ as our central focus. To avoid redundancy, we have relocated the entire section discussing the numerical study for the type-II functional response to Appendix\eqref{Appendixtype2}.
\section{Robustness Through Structural Sensitivity Analysis}\label{structures}
\begin{table}[h]
    \centering
    \begin{tabular}{|c|c|c|c|}
    \hline
    Functional Response & Formulation & $a$ & $b$ \\
    \hline
    Holling & $\displaystyle g_h(y) = \frac{a_h}{1 + b_h y}$ & 0.385 & 0.265 \\ & & &\\
    Ivlev & $\displaystyle g_i(y) = \frac{a_i}{b_i(1 - \exp(-b_i y))}$ & 0.21 & 0.15\\ & & & \\
    Trigonometric & $\displaystyle g_t(y) = \frac{a_t}{b_t}\tanh(b_t y)$ & 0.27 & 0.18  \\ & & & \\
    \hline
    \end{tabular}
    \caption{ Estimated parameters for different functional form}
    \label{tab2: estimated parameters for different functional forms}
\end{table}
Inspired by previous works \citep{flora2011structural, adamson2013can, adamson2014defining, adamson2020identifying}, we incorporated three functional responses—Monod, Ivlev, and Tan-Hyperbolic—in our study. We selected parameters for the Monod-type response based on \citep{wang2016modelling} and used a parameter set consistent with the type-II functional response, supporting stable coexistence. We adopted all values shown in Table\eqref{tab2: estimated parameters for different functional forms}, ensuring they meet condition (3.3) in \citep{adamson2013can}. Additional parameters are \( r = 0.6, a = 0.2, c = 0.42, m = 0.1, d = 0.32 \), aligning the response curves as in \citep{flora2011structural}.
\begin{figure}[h!]
\begin{center}
\includegraphics[height=3.0in,width=\textwidth]{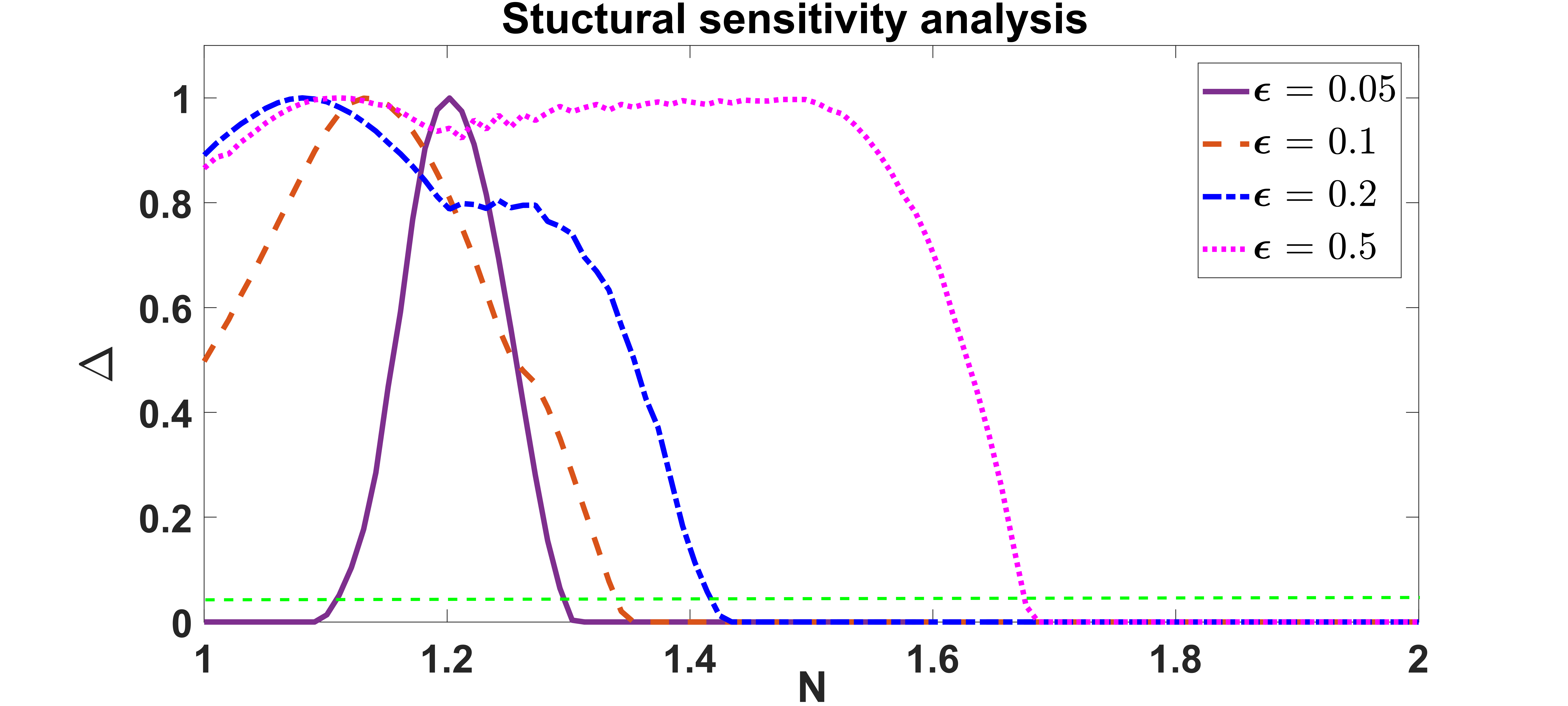}
    \end{center}
    \caption{Through this figure we explore the structural sensitivity of our model as the degree of partial cooperation ($N$) varies from $1$ to $2$, by choosing the Monod functional response as the base function of our study. Here three different $\epsilon$ values are 0.05(purple curve), 0.1(red curve), 0.2(blue curve) and 0.5 (pink curve) respectively. We use absolute distance for defining the $\epsilon_Q$ neighborhood. All the symbols and the definitions followed from the work \citep{adamson2013can}. Here x-axis represents the parameter $N$ and the y-axis represents the parameter $\Delta$. The green dashed line corresponds to the 'threshold sensitivity' of our system (at the level of $5\%$). That signifies that the degree of structural sensitivity below this margin is not sensitive.}
    \label{label fig 11}
    \end{figure}
Morozov. et.al \citep{adamson2013can} introduce the term 'degree of structural sensitivity'$(\Delta)$ to measure the likelihood that randomly selecting any two function sets from the $\epsilon_Q$-neighborhood independently would result in disparate predictions for the stability of the specified equilibrium and have analyzed this concept extensively with various examples in \citep{adamson2013can,adamson2014defining}. In a recent study \citep{sen2022bifurcation}, Morozov et al. investigated the robustness of the system dynamics through the lens of structural sensitivity.
\begin{figure}[h!]
\begin{center}
\includegraphics[height=3.5in,width=\textwidth]{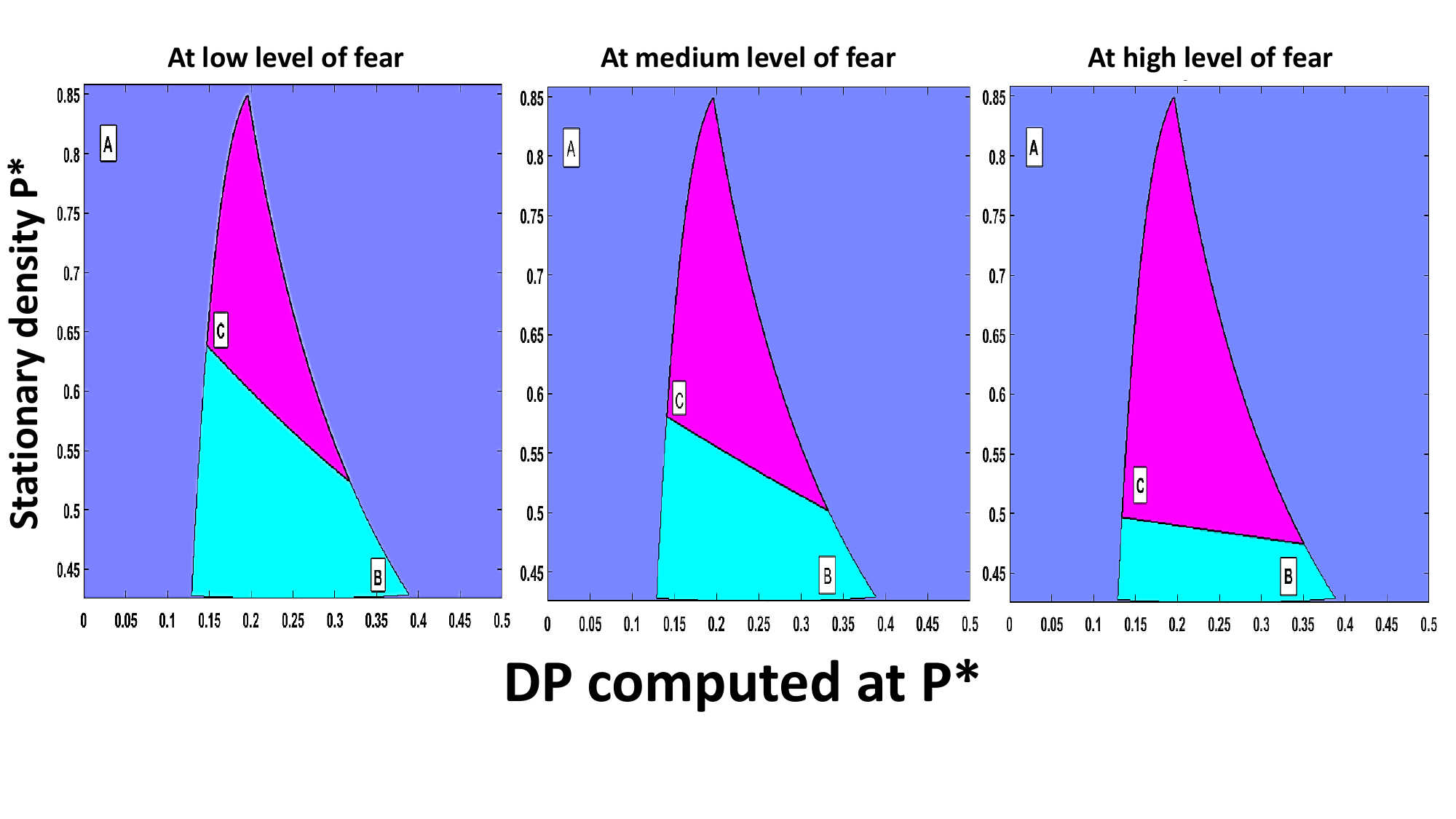}
    \end{center}
    \caption{By these three adjacent figures, we aim to capture how fear affects the sensitivity of our system. Here we plot the $p^*- \text{DP}$ plane, showing the area of feasible fixed point with possible DP value. The sky region(region \textbf{B}) shows the area of $p*$ values that are unstable and the pink shaded area (region \textbf{C}) for the region of stable fixed points. Whereas the blue region is the region of non-feasible points(region \textbf{A}). In this figure, the low level of fear means \(K=0.05\), at moderate fear, means \(K=0.5\) and we set \(K=5\) for high fear. }
    \label{label fig 12}
    \end{figure}
Having thoroughly examined the dynamics of our model, specifically the influence of fear on stability, we've noted compelling similarities between the dynamical behaviors of type-I and type-II functional responses when "$N$" is the bifurcation parameter. Now our objective is to evaluate the robustness of analytical behavior in light of these structural considerations. Our primary focus now lies in examining the sensitivity of our model dynamics to the mathematical structures assumed in Table\eqref{tab2: estimated parameters for different functional forms}, particularly across parameters "$N$" and "$K$". Additionally, to explore bi-parameter effects, we investigate sensitivity by varying the fear parameter "$K$" across the "$P^*-DP$" plane for a suitable chosen value of "$N$". Throughout, we maintain the definitions of variables $\Delta$, $\epsilon_Q$, $DP$, and $P^*$ consistent with those in  \citep{adamson2013can}. Figure\eqref{label fig 11} illustrates the degree of structural sensitivity ($\Delta$) for the system \eqref{main_model_1} over the parameter $N$ when the Monod type functional response is taken as the base function. In this representation, the parameter \(\Delta\) is plotted as the ordinate, while \(N\) is on the abscissa, creating an 'N-\(\Delta\)' space. The figure displays four distinct curves, each corresponding to different values of \(\epsilon\). The Euclidean distance rule, as described in \citep{adamson2013can}, is taken as the basis for this study. The solid violet curve represents the degree of structural sensitivity for \(\epsilon = 0.05\), the brown single dashed curve is for \(\epsilon = 0.1\), whereas the blue double dashed curve is for \(\epsilon = 0.2\) and at last the pink dotted curve corresponds to the curve for \(\epsilon = 0.5\). The other parameters are held constant at \(r = 0.6, K = 0.05, a = 0.2, a_h = 0.385, d = 0.373, c = 0.6, m = 0.1, b_h = 0.265\). As we observe the increase in "$N$", the system undergoes transitions from a stable region to an unstable one. Initially, when the perturbation value($\epsilon$) is very low, changes in the parameter"$N$" have minimal impact on stability at its initial values. However, upon surpassing a certain threshold, the value of $\Delta$ rapidly escalates, quickly reaching its maximum and inducing instability in the system. Interestingly, even a slight increase in "$\epsilon$" (from 0.05 to 0.1 ) reveals that structural sensitivity affects the system from its outset, highlighting its high sensitivity to initial conditions. Moreover, in this scenario, the system demonstrates enhanced resistance to destabilization caused by the parameter "$N$" for a longer duration than the previous one. As we increase the value of $\epsilon$, we observe that the value of $\Delta$ rises progressively, even at very initial values of "$N$". Concurrently, the range of resistance to destabilization expands accordingly. In Sen et al.'s work \citep{sen2022bifurcation}, the authors use the term "plasticity to resist destabilization" and demonstrate its presence in their system, a characteristic that holds true for our case as well. Furthermore, we can assert that our system exhibits high sensitivity to perturbations($\epsilon$) of the base function. We've also examined these results for the other two functional responses and found similar outcomes that we have got by taking type-II functional response form, indicating that the dynamical property with the bifurcation results of the system \eqref{main_model_general} is topologically robust to the functional response structure. Because, we have verified the basic bifurcation structure of our model for Type-I and Type-II responses, and the results indicate they remain topologically similar. However, we intentionally avoid detailing these findings to prevent repetition in the text. \\
Taking our investigation one step further, we are keen on understanding the impact of fear on the structural sensitivity of our model. To achieve this, we set a suitable value of \(N\) at $1.2$ and explore the fear effect at three different intensities: minimum fear (\(K = 0.05\)), medium fear (\(K = 0.5\)), and high-level fear (\(K = 5\)) respectively. The results are depicted in three figures, where the other parameters remain consistent with previous values. By analyzing these three figures (Figures \eqref{label fig 12}), our observations align with the earlier claim that fear indeed exhibits a stabilizing effect. As the value of the fear parameter \(K\) increases, the region of stability (pink area) expands concurrently, and the area containing unstable $P^*$ decreases simultaneously.
%In conclusion, the study of structural sensitivity provides valuable insights into how the stability of our system is influenced by changes in the functional response term, specifically shedding light on the effects of both the biological phenomenon of 'partial cooperation' and the 'fear effect'. We analyze the similarity of the effect of the other two functional response terms over the structural sensitivity but we skip that part here only for the space management purpose though someone can find the detailed review in the supplementary part.
\section{Results and Discussion}
While much of the research centers on the Allee effect in prey populations, often by explicitly incorporating the Allee threshold, recent studies have also begun to consider predator-driven Allee effects or Allee effect in predators by utilizing a special type of functional response \citep{kramer2010experimental,sen2022bifurcation}. However, in our study, the approach is different. In our work, the Allee effect is induced as a result of including the degree of partial cooperation. This approach is not only intriguing but also aligns closely with reality, as Allee effects often arise due to inter-specific or intra-specific positive interactions such as group defense or hunting cooperation or even from the fear of predation \citep{sasmal2022modeling}. Allee effect is not just an outcome of some social interaction but the level of social interaction within a species could indicate the degree to which it experiences Allee effects \citep{stephens1999consequences}. When examining Figure\eqref{label fig 1}, it becomes evident that the strong Allee effect is induced in our system and its threshold increases accordingly with the degree of partial cooperation. This observation aligns precisely with the assertion that the intensity of the Allee effect may be influenced by the extent of facilitation or degree of cooperation \citep{courchamp1999inverse}. This alignment between our dynamical results and existing biological evidence underscores the significance of this observation in our study.\\
Both under-crowding and overcrowding can pose limitations \citep{courchamp1999inverse}. In species that live in groups, if the group size becomes excessively large, there might not be adequate food to support all individuals. Conversely, if the group is too small, the advantages of cooperative foraging may not be fully realized \citep{rita1998group}. This emphasizes that the occurrence of the Allee effect or the maintenance of stability, despite the probability of extinction, is largely dependent on density in nature. Therefore, the persistence of a system despite adverse conditions is a myth; density-dependent extinction possibility is unavoidable. This assertion adds credibility to the existence of the extinction equilibrium as a stable state throughout the study due to the degree of partial cooperation-driven strong Allee effect, aligning well with existing literature. Consequently, this induces bistability in our system at each state $\mathcal{R}_1,\mathcal{R} _2, \mathcal{R}_3$ depicted in Figure \eqref{bifurcation_timeseries}, resulting in more intricate and compelling dynamics and enhancing the biological relevance of our model.\\
Despite potentially increasing the risk of extinction through the introduction of a strong Allee effect, cooperation fosters close associations among individuals, providing substantial benefits. In our model, prey species demonstrate enhanced fitness with increasing degrees of partial cooperation, leading to an elevation in their stable state density up to a certain threshold value of $N = N_{TB}$, in accordance with established biological principles. However, despite these increases, the available resources remain insufficient to meet the predator's requirements and as a result, regardless of the initial density, the predator population ultimately goes to extinction (see $\mathcal{R}_1$ in Figure\eqref{bifurcation_timeseries}). This suggests that the predator-free axial equilibrium not only maintains its stability but also shifts towards the right, as described in Figure \eqref{label fig 1}, thereby increasing the persistence capacity of prey density simultaneously. It demonstrated how cooperation can enhance competitors' abilities by augmenting their carrying capacities, which may underpin the existence of dynamic carrying capacity, aligning with existing biology \citep{zhang2003mutualism}. The concept of each region sustaining a specific level of species richness has been thoroughly explored in academic discourse \citep{rabosky2015species,cornell2013regional}. However, there are counterarguments against the idea of a fixed carrying capacity for species richness \citep{harmon2015species}. We propose that these challenges arise when carrying capacity is narrowly viewed as a strict upper limit on species numbers that an environment can accommodate. Instead, Storch et al. \citep{storch2019carrying} suggest it should be understood as a stable equilibrium arising from diversity-dependent dynamics of species richness, as we outline in our work (see in Section \eqref{model formulation}). These studies suggest that stable equilibria can occur even when not all resources are fully utilized. Despite the potential for increased pressure on resource consumption with higher population densities, species can employ cooperative strategies to optimize resource utilization and maintain themselves within their niche. This challenges the notion of species richness carrying capacity as a saturated niche space or hard ecological limit, instead supporting the concept of dynamic carrying capacity in a system. Our analysis successfully incorporates these cooperation-mediated dynamic results into our model through the concept of the degree of partial cooperation which remains largely unexplored.\\
In a predator-prey system, even with cooperation occurring only between two prey individuals, the predator population can't remain unaffected. In our case, if we increase $N$ further, immediately after crossing the threshold value $N_{TB}$, the coexistence equilibrium becomes stable as the increased prey density can now fulfill the predator's resource requirements. Consequently, with a further increase in the degree of partial cooperation for having sufficient resources, predators have to spend less energy and time searching for food, making predation easier. This ultimately increases the growth of the predator (as depicted in $\mathcal{R}_2$ of Figure \eqref{bifurcation_timeseries}). Interestingly, an increase in predator density usually leads to increased predation pressure in the system, resulting in a decrease in prey density. However, in our case, due to the higher degree of partial cooperation between prey species, it somehow helps sustain stable prey density at a certain constant level. Cooperation promoting coexistence in a predator-prey system is a well-known biological result \citep{zhang2003mutualism}, and our findings readily support this fundamental property of cooperation.\\
 In the context of a Holling type I functional response, typically associated with stabilizing effects, population oscillations observed in empirical studies are often linked to various factors such as spatial, seasonal, or environmental fluctuations. However, in our investigation, we found that surpassing the threshold value $N_{H}$ for the degree of cooperation triggers stable oscillations through a supercritical Hopf bifurcation, even with a type I functional response. The degree of partial cooperation serves as the underlying biological rationale. This increase in parameter value $N$ heightens predation pressure and augments carrying capacity, resulting in intensified inter-specific competition among prey species and potential destabilization of the system (Numerically established in Figure\eqref{label biparameter}). Interestingly, this increase in carrying capacity signifies the occurrence of the well-established biological phenomenon, the paradox of enrichment, which promotes oscillations within the system. These reasons can serve as a satisfactory underlying biological rationale for the discovery of oscillation in our system \eqref{main_model_general}.\\
If we further increase the value of $N$ beyond $N_{H}$, rapid growth of prey species followed by a rapid decline in prey population density due to increased predation occurs in the system. In a bounded environmental system, rapid predation of prey species can lead to over-exploitation, potentially causing a system collapse. This phenomenon is often associated with the presence of a heteroclinic orbit in predator-prey systems characterized by a strong Allee effect \citep{van2007heteroclinic}. Our study confirms the existence of such a heteroclinic orbit, indicating that when $N = N_{HB}$, the system undergoes a catastrophic collapse. This regime shift results in the merger of the two attractors, namely, the extinction equilibrium and oscillatory coexistence at the heteroclinic curve. In nature, as individuals strive to maximize the benefits of cooperation, their efforts frequently inadvertently create conditions that ultimately lead to its collapse \citep{stewart2014collapse}. In our case, the phenomenon of over-exploitation at higher degrees of partial cooperation serves as the possible biological rationale for this outcome. Indeed, in nature, the possibility of encountering uncertain and devastating situations is ever-present for any ecological system. Therefore, studying the maintenance of stability in the environment remains a crucial topic to explore. In our case, the degree of cooperation among prey species directly influences reproduction, and beyond a certain threshold, it leads to over-exploitation. In contrast, the fear of predation reduces reproduction rates, limiting prey availability for predators. Thus, fear of predation plays a critical role in preventing potentially catastrophic events like system collapse or species extinction by acting as a natural regulator. this leads to another significant finding of this study that fear can have a significant influence on the stability and persistence of the system for both type-I other than type-II functional responses, contrary to the claim that fear does not affect promoting stability in Wang et al.'s model for type-I functional response \citep{wang2016modelling}. Incorporating cost-associated cooperation,  a fundamental strategy in predator-prey interactions, provides a compelling rationale. Our result forecast that while fear typically resolves the 'paradox of enrichment' in predator-prey models by stabilizing oscillations \citep{wang2016modelling}, we show that under certain levels of cost-associated cooperation, the 'paradox of enrichment' can be reversed, even at higher level fear, leading to a novel dynamics that remain difficult to explore with type -I functional response. However, fear is not universally beneficial, as increasing fear levels reduce prey reproduction. Furthermore, if positive interactions among prey unintentionally increase predator density, the heightened fear of predation triggers a rapid decline in prey density, followed by a subsequent drop in predator density. This suggests that systems with higher predator densities and elevated fear levels are more prone to extinction. This rarely explored phenomenon underscores the destabilizing nature of fear, as depicted by the basin of attraction in Figure \eqref{fear basin}, which captures the system's vulnerability to collapse. At higher predation rates, selective or cognitive cooperative strategies among prey can become costly to the system. Our bi-parameter studies confirm that even minimal partial cooperation destabilizes the system under high predation. However, when fear of predation is present, increased partial cooperation amplifies anti-predator defenses, emphasizing the dilemma of cost-associated cooperation and reinforcing the biological significance of our findings. Structural sensitivity analysis further indicates that these rich dynamical results remain robust across functional response structures.\\
This study also aims to explain the rationale behind the reproduction term proposed by Hatton et al. \citep{hatton2015predator} for species like mammals, primates, or other terrestrial animals, where cost-associated partial cooperation is more commonly observed. The strength of this study lies in highlighting the significance of the mean-field power law approach in modeling ecological dynamics such as demographic Allee effects, regime shifts, and catastrophic collapses—typically explored through more complex models. As most biological processes, including cooperation, are density-dependent, the framework aligns with real-world scenarios. For example, like hunting cooperation \citep{alves2017hunting}, in our case partial cooperation is advantageous up to a certain point (regions $\mathcal{R}_1$ and $\mathcal{R}_2$), but beyond this, it can lead to system collapse (region $\mathcal{R}_4$). This study provides key mathematical insights into identifying and analyzing equilibrium points in mean-field models with transcendental equations. Our examination of heteroclinic orbits captures dynamics around saddle nodes, offering a novel approach for predator-prey modeling. Additionally, it clearly distinguishes heteroclinic bifurcations from homoclinic ones, which remain volatile in numerical simulations. Investigating the probability of extinction in such systems would be valuable \citep{rakshit2024regime}. Additionally, exploring a multi-connected network with associated costs of cooperation having special fluctuation in different patches presents an intriguing avenue for further research.  
%\section*{Author Contributions:}  
%S.C. was responsible for methodology development, formal analysis, software implementation, drafting the original manuscript, and contributing to review and editing. S.S. was involved in conceptualization, methodology, supervision, and manuscript review and editing. J.C. provided supervision, conceptualization, and manuscript review and editing.
\section*{Data Availability}
The datasets used and/or analysed during the current study available
from the corresponding author on reasonable request.
\appendix
\section*{Appendix}
\section{ Uniform Persistence}
Here, we discuss the uniform persistence of our main model. To maintain biological relevance in our discussion, we define uniform persistence only in the positive quadrant of the $\mathbb{R}^2$ plane. Firstly, we'll demonstrate the positivity of our system, followed by a discussion on its boundedness.
 \subsubsection*{ Positiveness of the system}
To prove the positiveness of the system, we take the first equation and integrate both sides we have, 
$$
\begin{array}{l}
\int_0^t\frac{d x}{x}=\int_0^t \left(\frac{r_0 x(s)^N}{1+k y(s)}-a x(s)-p y(s)-d\right) d s \\
x(t)=x(0) \exp \left[\int _ { 0 } ^ { t } \left(\frac{r x^N(s)}{1+k y(s)}-a x(s)-p y(s)-d\right)ds\right].
\end{array}
$$
Similarly, if we take the second equation of the system [ref number]
$$
\begin{array}{l}
\int_0^t \frac{d y}{d y}=\int_0^t(q x(s)-m) d s \Rightarrow
y(t)=y(0) \exp \left[\int_0^t(q x(s)-m) d s\right] .
\end{array}
$$
For the biological rationale as we will take, $x(0) \& y(0)>0$, then from the previous discussion it is clear that $x(t) \& y(t)>0$.
\subsubsection*{ Boundedness of the system:}
Let us first assume that,$ z = u+v $ and $ \alpha $ be any positive real number. Also, we have assumed that $1 \leqslant N<2$ in our model and also as  $c<1$ then $q<p$. Now,
\begin{itemize}
\item{Case-1:}
if $u \leq 1,$ then $u^N\leq1$. which implies that, 
\begin{align*}
\frac{d z}{d t}+\alpha z&=\frac{du}{d t} +\frac{d v}{ds}+\alpha(u+v)
= \frac{r_0 u^N}{1+k v}-a u^2 + p u v -d u+\alpha u+q u v-m v+\alpha v ,\\
\leq& r_0-a u^2+\alpha u-(p-q) u v-(m-\alpha) v ,
\leq r_0+a u\left(\frac{\alpha}{a}-u\right)-(m-\alpha) v , \\
\leq& r_0+B-(m-\alpha) v .
\end{align*}
Where $B=\frac{\alpha^2}{4 a}$ and also for a suitable chosen value of $\alpha>0$, we can say that, $\frac{d z}{d t}+\alpha z$ will be bounded.\\
\item{Case-2:}
If \( u > 1 \), then \( u^N < u^2 \). Therefore, if \( r_0 \leq a \), then \(\frac{du}{dt} < 0\) when \( d > 0 \), indicating that the prey species will go extinct, which is not our concern. Hence, we will choose \( r_0 > a \).
Now,
\begin{align*}
\frac{d z}{d t}+\alpha z & =\frac{r_0 u^N}{1+k v}-a u^2-p u v-d u+\alpha u+q u v - mv+\alpha v,\\
& \leq r_0 u^2-a u^2+\alpha u+(\alpha -m) v 
 =\left(r_0-a\right) u^2+\alpha u+(\alpha-m) v,\\
&=\left(r_0-a\right) u\left(\frac{\alpha}{r_0-a}-u\right)+(\alpha-m) v \leq M_2+(\alpha-m) v.
\end{align*}
Where, $M_2=\frac{\alpha^2}{\left|a-r_0\right|}$ and also for a suitable choice of $\alpha$, $\frac{d z}{d t}+\alpha z \leq M_2$.
So, by analyzing both cases we can conclude, that the system will remain bounded (as for $\alpha = m$ both the cases hold together).\\
\end{itemize}
\section{ Transcritical Bifurcation}\label{APP_trans}
Let, $x=\left(\begin{array}{l}u \\ v\end{array}\right)$ and $f(x)=D\left(\begin{array}{l}u \\ v\end{array}\right)$ where, $D$ is the time derivative.\\
Hence, $f(x)=\left[\begin{array}{c}\frac{r_0 u^N}{1+K v}-a u^2-d u-p u v \\ c p u v-m v\end{array}\right]$
$$
J(\beta, 0)=\left[\begin{array}{cc}
N r_0 \beta^{N-1}-2 a \beta-d & -K r_0 \beta^N-p\beta \\
0 & q \beta-m
\end{array}\right]
$$
Now, we are going to check the bifurcation at $m_0=m=q \beta$
$$
\begin{array}{l}
J(P_1, m_0)=\left[\begin{array}{cc}
\left(N r_0 \beta^{N-1}-2 a \beta-d\right) & \left(-Kr_0 \beta^N-p \beta\right) \\
0 & 0
\end{array}\right] \\
\text {  }  \text {  }  \text { }
\end{array}
$$
Now if $v' = (v_1,v_2)^{T}$ and $w' = (w_1,w_2)^T$ are the eigenvectors corresponding to the null eigenvalue of the matrices $J_m(B, 0)$ and $\left[J_m(B, 0)\right]^{T}$ respectively.Then, $v'=\left(-p_2 v_2 / p_1, v_2\right)$ and
$w'=\left(0, w_2\right) \text {. }$
Again if we take the partial derivative of $f(x)$ concerning the parameter $m$, we have
$$
\left.f_m(x)\right|_{(\beta, 0)}=\left(\begin{array}{c}
0 \\
-v
\end{array}\right)=\left(\begin{array}{l}
0 \\
0
\end{array}\right)
$$
Hence, $w'^{T} f_{\left(P_1, m_0\right)}=0$.
Again if we take 
$
D f_{\left(P_1,m_0\right)}(x)=\left(\begin{array}{cc}
0 & 0 \\
0 & -1
\end{array}.\right)
$
So,
$
w'^{T}\left[D f_{(P_1, m_0)} \cdot v'\right]=-w_2 v_2 \neq 0 .
$
Now we calculate 
$D^2 f(x)=\left(\begin{array}{cccc} S_1 & S_2 & S_3 & S_4 \\ 0 & q & q & 0\end{array}\right)$
where,
$$
\begin{array}{l}
S_1=N(N-1) r_0 \beta^{N-2}-2 a , 
S_2=-K r_0 N \beta^{N-1}-p , 
S_3=-K r_0 N \beta^{N-1}-p , 
S_4=-2 K^2 r_0 \beta^N
\end{array}
$$
Hence, $w'^{T} D^2 f(x)_{(\beta, 0)}(v', v')=2 q w_2 v_1 v_2 \neq 0$.
\section{Calculations for Lyapunov exponent}\label{lyapunov}
\subsection*{Compute Second-Order Partial Derivatives }
\begin{align*}
Q_1^{11} &= \frac{\partial^2 Q_1}{\partial y_1^2} \bigg|_{(0, 0; N_H)}
= \frac{\partial}{\partial y_1} \left[ \left( c_{20} c_{01}^2 - c_{11} c_{01} c_{10} + c_{02} c_{10}^2 \right) y_1 + \lambda_H \left( 2 c_{02} c_{10} - c_{11} c_{01} \right) y_2 \right] \bigg|_{(0, 0; N_H)}\\ 
&= \left( c_{20} c_{01}^2 - c_{11} c_{01} c_{10} + c_{02} c_{10}^2 \right),\\
Q_1^{22} &= \frac{\partial^2 Q_1}{\partial y_2^2} \bigg|_{(0, 0; N_H)}= \frac{\partial}{\partial y_2} \left[ \lambda_H \left( 2 c_{02} c_{10} - c_{11} c_{01} \right) y_1 + \lambda_H^2 c_{02} y_2 \right] \bigg|_{(0, 0; N_H)} = \lambda_H^2 c_{02},\\
Q_1^{12} &= \frac{\partial^2 Q_1}{\partial y_1 \partial y_2} \bigg|_{(0, 0; N_H)} = \frac{\partial}{\partial y_2} \left[ \left( 2 c_{02} c_{10} - c_{11} c_{01} \right) y_1 + \lambda_H \left( c_{12} c_{01} c_{10}^2 - c_{03} c_{10}^3 + c_{30} c_{01}^3 - c_{21} c_{01}^2 c_{10} \right) y_1^2 \right] \bigg|_{(0, 0; N_H)} \\
&= \lambda_H \left( 2 c_{02} c_{10} - c_{11} c_{01} \right),\\
%\subsection{Second-Order Partial Derivatives of \( Q_2 \):}
Q_2^{11} &= \frac{\partial^2 Q_2}{\partial y_1^2} \bigg|_{(0, 0; N_H)} = \frac{\partial}{\partial y_1} \left[ \left( d_{20} c_{01}^2 - d_{11} c_{01} c_{10} + d_{02} c_{10}^2 \right) y_1 + \lambda_H \left( 2 d_{02} c_{10} - d_{11} c_{01} \right) y_2 \right] \bigg|_{(0, 0; N_H)} \\
&= \left( d_{20} c_{01}^2 - d_{11} c_{01} c_{10} + d_{02} c_{10}^2 \right),\\
Q_2^{22} &= \frac{\partial^2 Q_2}{\partial y_2^2} \bigg|_{(0, 0; N_H)} = \frac{\partial}{\partial y_2} \left[ \lambda_H \left( 2 d_{02} c_{10} - d_{11} c_{01} \right) y_1 + \lambda_H^2 d_{02} y_2 \right] \bigg|_{(0, 0; N_H)} = \lambda_H^2 d_{02},\\
Q_2^{12} &= \frac{\partial^2 Q_2}{\partial y_1 \partial y_2} \bigg|_{(0, 0; N_H)} = \frac{\partial}{\partial y_2} \left[ \left( 2 d_{02} c_{10} - d_{11} c_{01} \right) y_1 + \lambda_H \left( d_{12} c_{01} c_{10}^2 - d_{03} c_{10}^3 + d_{30} c_{01}^3 - d_{21} c_{01}^2 c_{10} \right) y_1^2 \right] \bigg|_{(0, 0; N_H)} \\
&= \lambda_H \left( 2 d_{02} c_{10} - d_{11} c_{01} \right).
\end{align*}
\subsection*{Compute the Third-Order Partial Derivatives}
\begin{align*}
Q_1^{111} &= \frac{\partial^3 Q_1}{\partial y_1^3} \bigg|_{(0, 0; N_H)}= \frac{\partial}{\partial y_1} \left[ c_{20} c_{01}^2 - c_{11} c_{01} c_{10} + c_{02} c_{10}^2 \right] \bigg|_{(0, 0; N_H)} = 0,\\
Q_1^{122} &= \frac{\partial^3 Q_1}{\partial y_1 \partial y_2^2} \bigg|_{(0, 0; N_H)} = \frac{\partial}{\partial y_2} \left[ \lambda_H \left( 2 c_{02} c_{10} - c_{11} c_{01} \right) \right] \bigg|_{(0, 0; N_H)} = 0,\\
Q_1^{212} &= \frac{\partial^3 Q_1}{\partial y_2 \partial y_1^2} \bigg|_{(0, 0; N_H)} = \frac{\partial}{\partial y_2} \left[ \left( c_{12} c_{01} c_{10}^2 - c_{03} c_{10}^3 + c_{30} c_{01}^3 - c_{21} c_{01}^2 c_{10} \right) y_1^2 \right] \bigg|_{(0, 0; N_H)} = 0,\\
Q_1^{222} &= \frac{\partial^3 Q_1}{\partial y_2^3} \bigg|_{(0, 0; N_H)} = \frac{\partial}{\partial y_2} \left[ \lambda_H^2 c_{02} \right] \bigg|_{(0, 0; N_H)}= 0,\\
%\subsection*{Third-Order Partial Derivatives of \( Q_2 \):}
Q_2^{111} &= \frac{\partial^3 Q_2}{\partial y_1^3} \bigg|_{(0, 0; N_H)}= \frac{\partial}{\partial y_1} \left[ d_{20} c_{01}^2 - d_{11} c_{01} c_{10} + d_{02} c_{10}^2 \right] \bigg|_{(0, 0; N_H)} = 0,\\
Q_2^{122} &= \frac{\partial^3 Q_2}{\partial y_1 \partial y_2^2} \bigg|_{(0, 0; N_H)} = \frac{\partial}{\partial y_2} \left[ \lambda_H \left( 2 d_{02} c_{10} - d_{11} c_{01} \right) \right] \bigg|_{(0, 0; N_H)} = 0,\\
Q_2^{212} &= \frac{\partial^3 Q_2}{\partial y_2 \partial y_1^2} \bigg|_{(0, 0; N_H)}= \frac{\partial}{\partial y_2} \left[ \left( d_{12} c_{01} c_{10}^2 - d_{03} c_{10}^3 + d_{30} c_{01}^3 - d_{21} c_{01}^2 c_{10} \right) y_1^2 \right] \bigg|_{(0, 0; N_H)} = 0,\\
Q_2^{222} &= \frac{\partial^3 Q_2}{\partial y_2^3} \bigg|_{(0, 0; N_H)}=\frac{\partial}{\partial y_2} \left[ \lambda_H^2 d_{02} \right] \bigg|_{(0, 0; N_H)}= 0.
\end{align*}
\section{Simulation with type-II functional response}\label{Appendixtype2}
\begin{figure}[h!]
\centering
 \subfloat[\centering shows the bistability between the non-saddle axial equilibrium point and the extinction equilibrium for our parameter set.]
{\includegraphics[height=4.3cm,width=5.7cm]{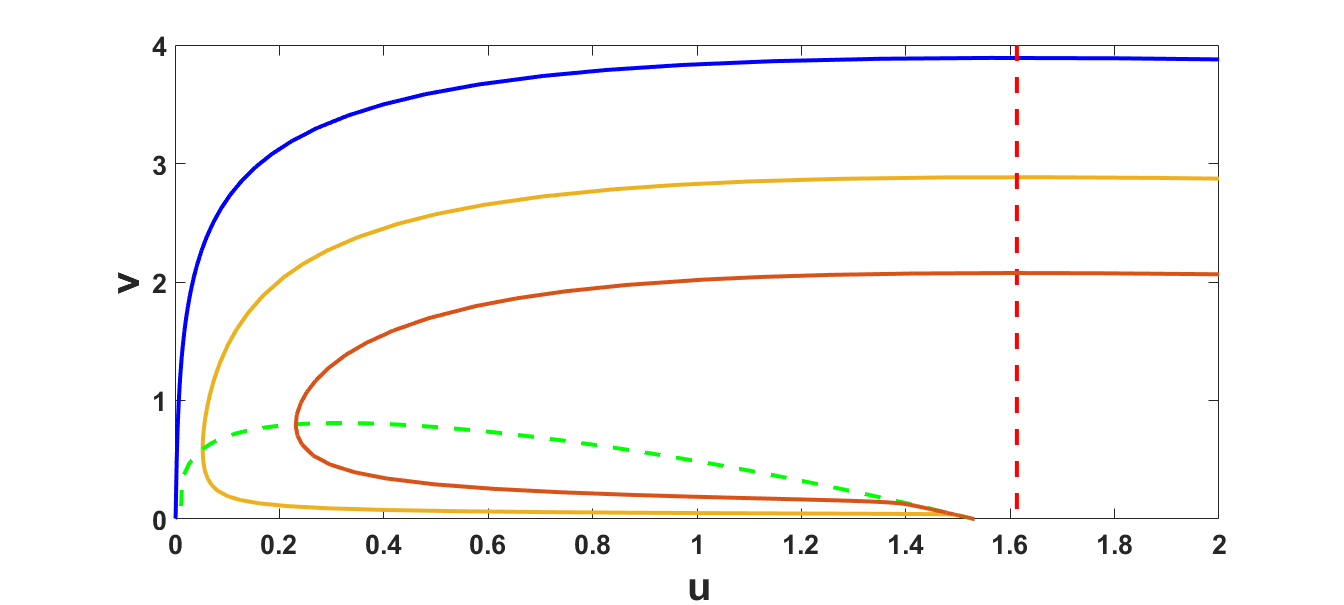}} 
\quad
   \subfloat[\centering The maroon and yellow curves asymptotically converge to the interior fixed point and the blue curve goes to extinction. The parameter values and other details are in the description.]{\includegraphics[height=4.3cm,width=5.7cm]{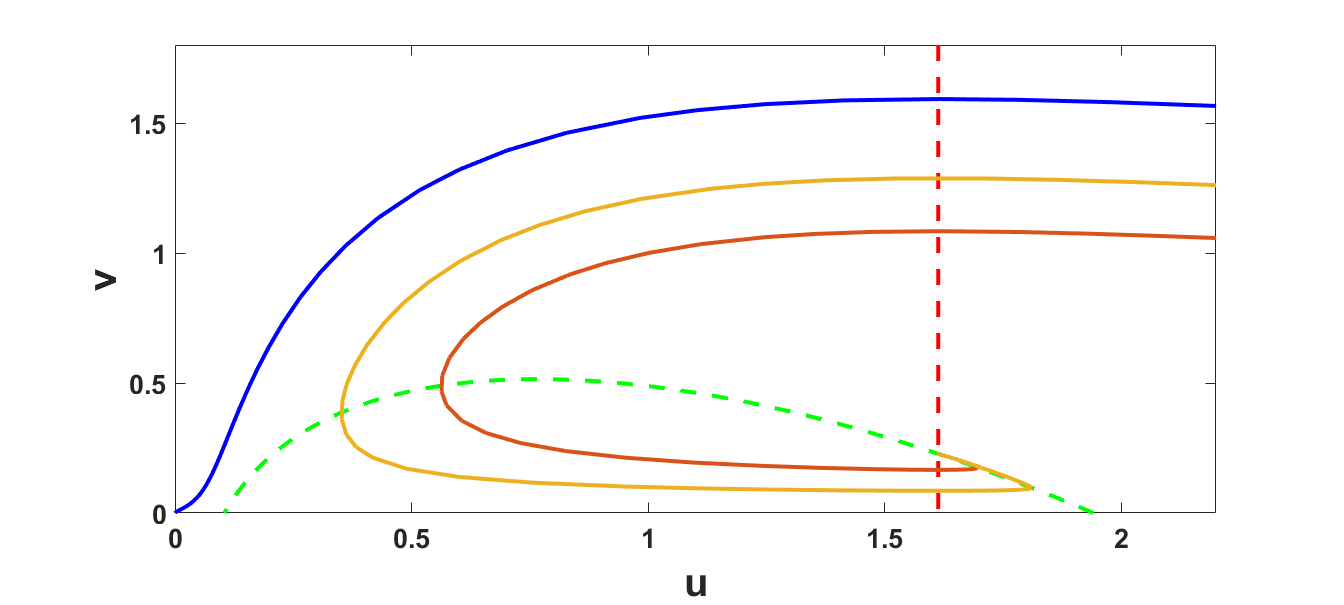}}
    \caption{correspond to the change of stability of axial equilibrium with the increasing value of N}%
    \label{fig:example1}
     \label{label fig 3}
    \end{figure}
In this section, we explore both type-I and type-II functional responses, focusing primarily on type-II dynamics while acknowledging their topological equivalence in bifurcation structure over \( N \). By plotting prey and predator nullclines, we highlight equilibrium points and bistability for both responses. Green dotted lines represent prey nullclines, while red vertical lines depict predator nullclines. For simplicity, we use \( r \) instead of \( r_0 \) in the simulations. The parameter values utilized for this analysis are as follows: \(r = 0.6\), \(N = 1.1\), \(K = 0.0\), \(a = 0.2\), \(a_h = 0.385\), \(d = 0.32\), \(c = 0.42\), and \(m = 0.1\) and  \(b_h = 0.265\).
\begin{itemize}
    \item Study the effect of partial cooperation without fear effect:
\end{itemize}
\begin{figure}[h!]
\centering
 \subfloat[\centering shows the bistability of the system between the stable limit cycle and the extinction equilibrium point when \(N = 1.415\).]
{\includegraphics[height=4.3cm,width=4.1cm]{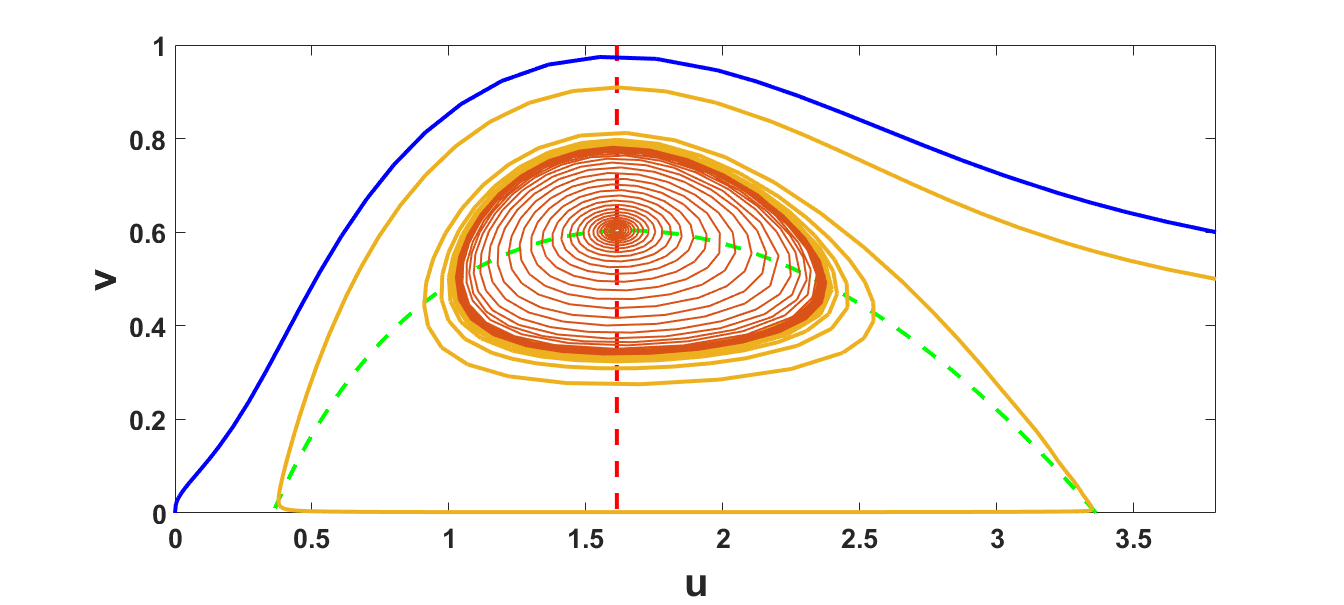}} 
\quad
   \subfloat[\centering When $N=1.4255$ situation is very similar to the previous figure except the amplitude of the oscillation increases.]{\includegraphics[height=4.3cm,width=4.1cm]{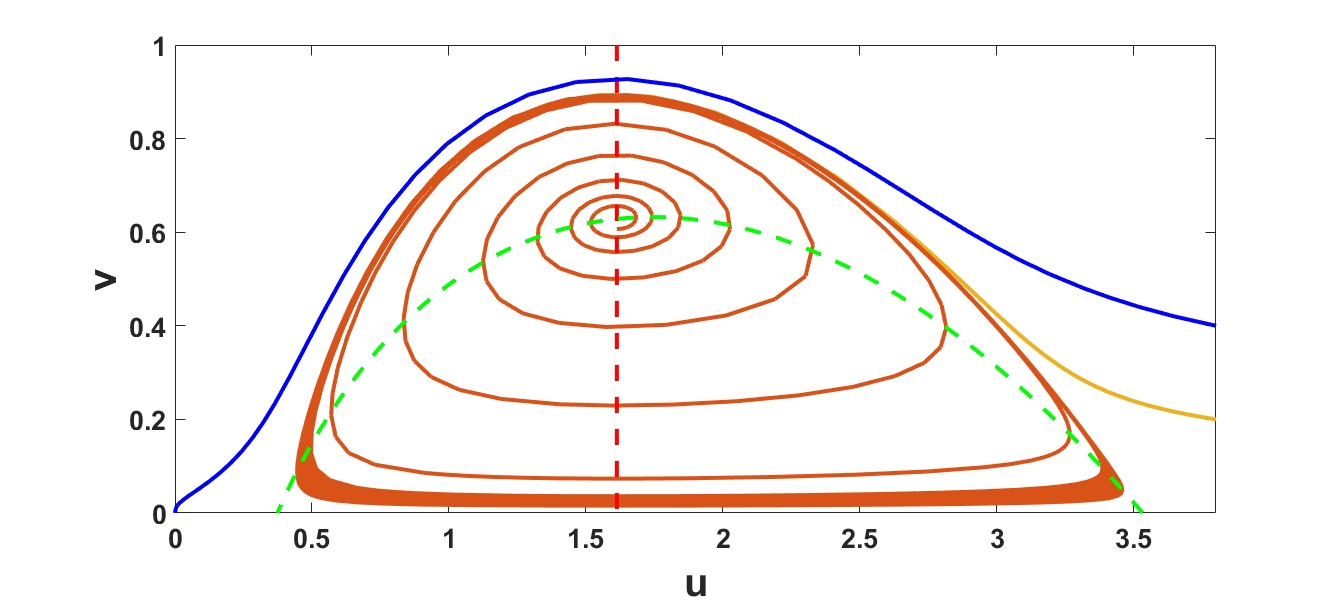}}
   \quad
   \subfloat[\centering Now the value of N increased to 1.45 but the rest of the parameters remain the same but all the trajectories ultimately go to the origin.]{\includegraphics[height=4.3cm,width=4.1cm]{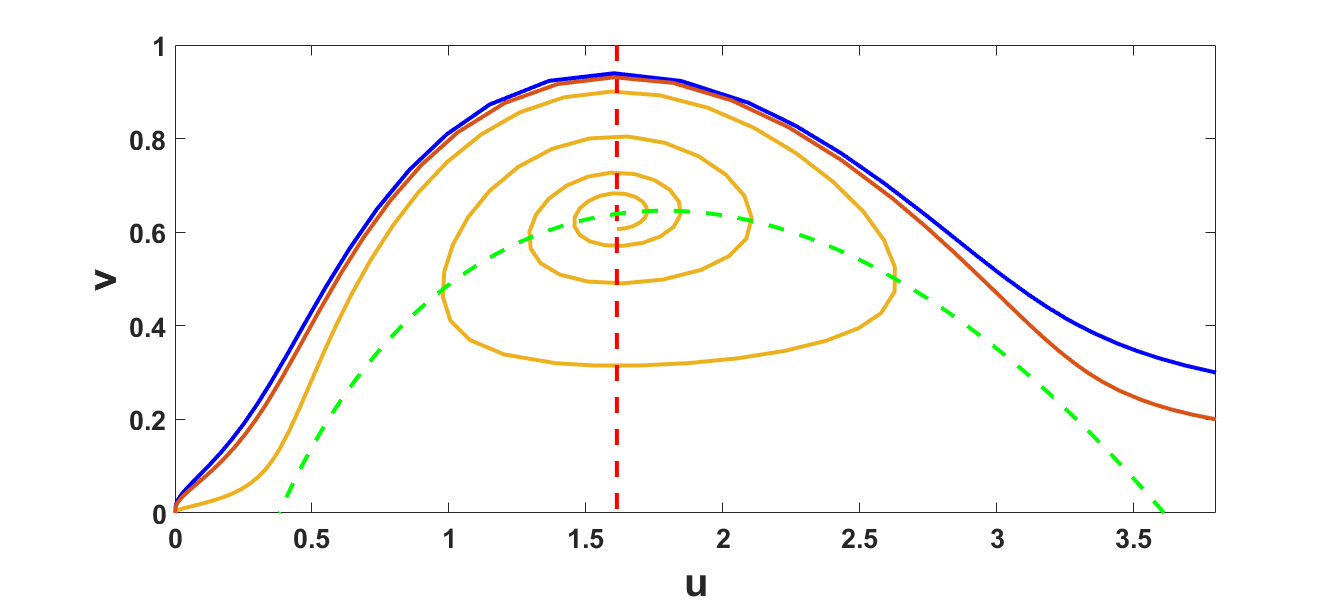}}
    \caption{correspond to the change of stability of interior equilibrium with the increasing value of N}%
    \label{fig:example2}
     \label{label fig 5}
\end{figure}
Figure \eqref{label fig 3}(a) shows two axial equilibrium points: a saddle-node at the left intersection of the prey nullcline with the u-axis and a stable node. The origin is a stable fixed point, and no interior equilibrium exists since the prey and predator nullclines do not intersect. The system exhibits bistability between the origin and the predator-free equilibrium. The blue trajectory from $(3,3.8)$ converges to the origin, while the yellow and maroon trajectories from $(3,2.8)$ and $(3,2)$ converge to the predator-free equilibrium. In Figure \eqref{label fig 3}(b), with \(N = 1.25\), both axial equilibria become saddle nodes, and the nullclines intersect, forming a stable interior equilibrium via a transcritical bifurcation at \(N = 1.14342580\). Bistability now exists between the origin and the coexisting equilibrium, enriching ecosystem dynamics. However, small perturbations in initial conditions could lead to extinction as we see the interesting fact that the initial point of the blue curve and the yellow curve are very near to each other($(3.1.5)$and $(3,1.25)$ respectively), adding complexity to the model. In Figures \eqref{label fig 5}(a) and (b), the stable interior node undergoes a supercritical Hopf bifurcation at \(N = 1.41046037\). As \(N\) increases, the oscillation amplitude grows, and the limit cycle eventually collides with the separatrix, disappearing through a heteroclinic bifurcation at \(N = 1.42623326\) but bi-stability is noted. Afterward, the origin becomes the sole attractor, resulting in global stability at the origin.
\begin{figure}[h!]
\begin{center}
\includegraphics[height=2in,width=\textwidth]{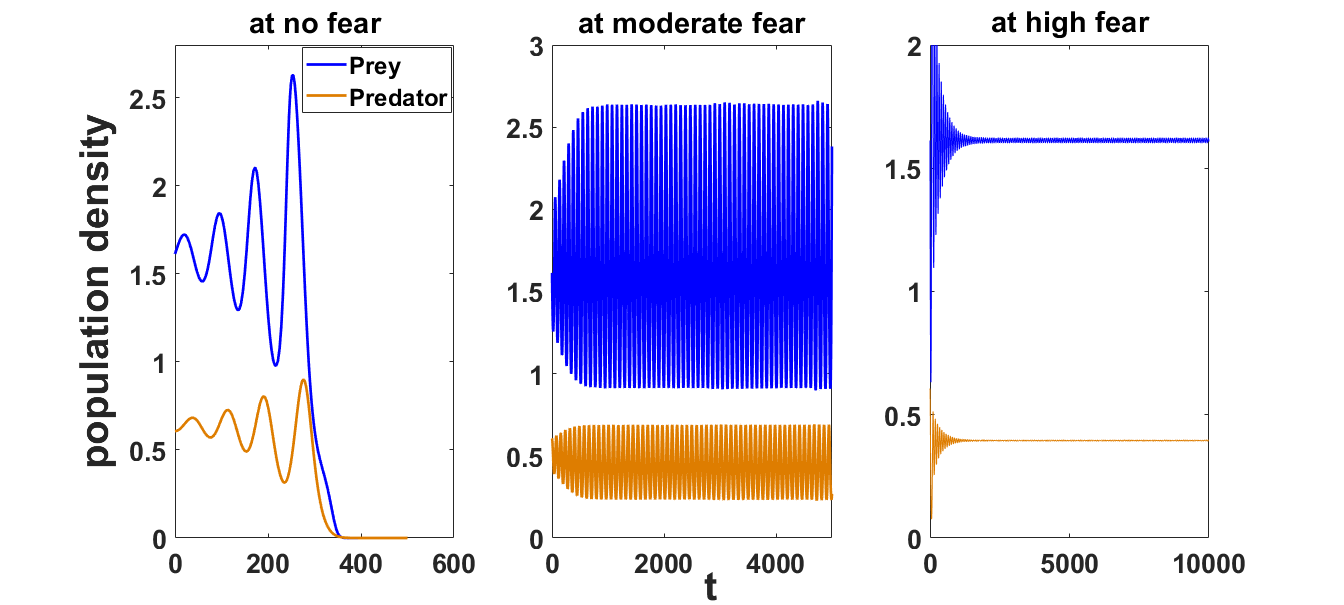}
    \end{center}
    \caption{Here we want to explore the effect of fear for different levels of fear. We take N=1.43 and the rest of the parameters are as before and we have taken K=0.00 as case(a) for the case of low fear level, K=0.05 as case(b) for the case of moderate fear, and K=0.15 as case (c) to show the high level of fear.}
    \label{label fig 7}
    \end{figure}
\begin{itemize}
    \item Study the effect of fear on the dynamics of the system:
\end{itemize}
Our primary focus now is to explore the impact of fear, and from the analysis of Figure\eqref{label fig 7}, we can infer that fear exerts a stabilizing effect on our system. In Figure \eqref{label fig 7}(a), where \(K = 0\), the origin is the sole attractor and is globally stable. Moving on to Figure \eqref{label fig 7}(b) with \(K = 0.05\), while the interior remains unstable, a stable limit cycle emerges around the fixed point. Subsequently, in \eqref{label fig 7}(c) with \(K = 0.15\), the interior equilibrium regains stability, and as the value of fear increases, the interior equilibrium maintains its stability, although the predator density decreases continuously at the stable equilibrium. It is noteworthy that the revert to a transcritical bifurcation is not possible, regardless of the increased value of the parameter 'K’.
\bibliography{SR}
\section*{Author contributions statement}
S.C. was responsible for methodology development, formal analysis, software implementation, drafting the original manuscript, and contributing to review and editing. S.S. was involved in conceptualization, methodology, supervision, and manuscript review and editing. J.C. provided supervision, conceptualization, and manuscript review and editing.
\section*{Additional information}
\textbf{Competing interests:} The authors declare that there are no financial or personal relationships that could have influenced the research presented in this paper.
\end{document}